\documentclass[pra,twoside,amsmath,amsfonts,amssymb,superscriptaddress,showpacs]{revtex4}
\usepackage{amsmath}
\usepackage{amscd}
\usepackage{float}
\usepackage{longtable}
\usepackage{epsfig}
\usepackage{amssymb}
\usepackage{amsfonts}
\usepackage{latexsym}
\usepackage{theorem}
\usepackage{times}
\usepackage{bm}
\usepackage[pdftex]{color}
\usepackage[cmtip,arrow]{xy}
\usepackage{pb-diagram,pb-xy}
\usepackage{verbatim}
\usepackage{fancyhdr}
\usepackage{tikz}
\usepackage{amscd}
\newtheorem{thm}{Theorem}
\newtheorem{defi}[thm]{Definiton}
\newtheorem{lem}[thm]{Lemma}
\newtheorem{prop}[thm]{Proposition}
\newtheorem{cor}[thm]{Corollary}

\newtheorem{rem}[thm]{Remark}

\newenvironment{proof}[1][{}]{\par\vskip12ptplus30pt\noindent{\it Proof{#1}:\ }}%
{\hfill$\blacksquare$\par\vskip12ptplus60pt}
\def\1#1{{\bf #1}}
\def\2#1{{\mathcal #1}}
\def\4#1{{\tt #1}}
\def\5#1{{\sf #1}}
\def\6#1{{\mathfrak #1}}
\def\7#1{{\Bbb #1}}
\def\8#1{{\rm #1}}
\def\9#1{{\mathcurl #1}}
\def\bS#1{{\boldsymbol #1}}
\DeclareFontFamily{OT1}{rsfs}{}
\DeclareFontShape{OT1}{rsfs}{m}{n}{<-7> rsfs5 <7-10> rsfs7 <10-> rsfs10}{}
\DeclareMathAlphabet\mathcurl{OT1}{rsfs}{m}{n}
\definecolor{grey}{rgb}{0.5,0.5,0.5}

\def\I{\openone}
\def\tr{{\rm tr}}
\newcommand{\CAR}[2]{\9F(#1,#2)}
\newcommand{\AFOCK}[1]{F_-(#1)}
\newcommand{\SPROD}[2]{\langle{#1},{#2}\rangle}
\newcommand{\RIGG}[3]{{\langle}{#1},{#2}{\rangle}_{#3}}

\newcommand{\GR}[2]{\Lambda(#1,#2)}

\newcommand{\GAR}[3]{\9G(#1,#3;#2)}

\newcommand{\IDL}[3]{\9I(#1,#3;#2)}

\newcommand{\PGAR}[3]{\mathring{\9G}(#1,#3;#2)}
\newcommand{\GCO}[2]{[#1,#2]_{\rm g}}
\newcommand{\CO}[2]{[#1,#2]}
\newcommand{\ACO}[2]{\{#1,#2\}}

\newcommand{\PHSP}[3]{\9R(#1,#3;#2)}

\newcommand{\GHOL}[4]{\9O(#1,#3,#2|#4)}
\newcommand{\RHOL}[3]{\9O(#1,#3,#2)}
\newcommand{\RHOM}[4]{{\rm Hom}(#1,#3,#2|#4)}
\newcommand{\REND}[3]{{\rm End}(#1,#3,#2)}
\newcommand{\ATEN}[2]{\mathring{\Lambda}(#1,#2)}
\newcommand{\HNO}[1]{\|\hspace{-1pt}| #1 \|\hspace{-1pt}|}
\newcommand{\FOUR}{{\9F}}
\newcommand{\CON}{{\bS\ast}}
\newcommand{\ADJ}{\star}
\newcommand{\RGHOM}[3]{{\rm Hom}(#1,#3,#2)}
\newcommand{\PCN}[1]{\|#1\|_\gamma}
\newcommand{\ENV}[3]{\9E(#1,#3;#2)}
\newcommand{\EQ}[2]{\left[#1\right]_{#2}}
\tikzstyle{gbstyle}= [draw=gray!80, fill=gray!80, thick, minimum height=.4cm, text centered,%
text=white]
\tikzstyle{ghstyle}= [draw=gray!20, fill=gray!20, thick, minimum height=.4cm, text centered]
\tikzstyle{pcstyle}= [thick, shape=circle, draw=black, minimum size = 1.5cm, text centered]
\tikzstyle{tbox} = [draw, fill=white, text width=2cm, text centered, minimum height=2cm]

\begin{document}
%%%%%%%%%%%%%%%%%%%%%%%%%%%%%%%%%%%%%%%%%%%%%%%%%%%%%%%%%%%%%%%%%%%%%%%%%%%%%%%%%%%%%%%%%%%%%%%%%%%%
\title{The algebra of Grassmann canonical anti-commutation relations (GAR) and its applications to
fermionic systems}

\author{Michael Keyl}
\email{m.keyl@tu-bs.de}

\affiliation{ISI Foundation, Quantum Information Theory Unit,\\ Viale S. Severo 65, 10133 Torino,
Italy}

\affiliation{Institut f\"ur
Mathematische Physik, Technische Universit\"at Braunschweig,
Mendelssohnstra{\ss}e~3, 38106 Braunschweig, Germany}

\author{Dirk-M. Schlingemann}
\email{d.schlingemann@tu-bs.de}

\affiliation{ISI Foundation, Quantum Information Theory Unit,\\
Viale S. Severo 65, 10133 Torino, Italy}

\affiliation{Institut f\"ur Mathematische Physik, Technische
Universit\"at Braunschweig, Mendelssohnstra{\ss}e~3, 38106 Braunschweig, Germany}

\pacs{03.67.-a, 02.30.Tb}

\date{\today}

%%%%%%%%%%%%%%%%%%%%%%%%%%%%%%%%%%%%%%%%%%%%%%%%%%%%%%%%%%%%%%%%%%%%%%%%%%%%%%%%%%%%%%%%%%%%%%%%%%%%
\begin{abstract}
We present an approach to a non-commutative-like phase space which allows to analyze quasi-free
states on the CAR algebra in analogy to quasi-free states on the CCR algebra. The used mathematical
tools are based on a new algebraic structure the ``Grassmann algebra of canonical anti-commutation
relations'' (GAR algebra) which
is given by the twisted tensor product of a Grassmann and a CAR algebra. As a new application, the
corresponding theory provides an elegant tool for calculating the fidelity of two quasi-free
fermionic states which is needed for the study of entanglement distillation within fermionic
systems.
\end{abstract}
%%%%%%%%%%%%%%%%%%%%%%%%%%%%%%%%%%%%%%%%%%%%%%%%%%%%%%%%%%%%%%%%%%%%%%%%%%%%%%%%%%%%%%%%%%%%%%%%%%%%
\maketitle
\pagestyle{fancy}
%%%%%%%%%%%%%%%%%%%%%%%%%%%%%%%%%%%%%%%%%%%%%%%%%%%%%%%%%%%%%%%%%%%%%%%%%%%%%%%%%%%%%%%%%%%%%%%%%%%%
\section{Introduction}
\label{sec:intro}
%%%%%%%%%%%%%%%%%%%%%%%%%%%%%%%%%%%%%%%%%%%%%%%%%%%%%%%%%%%%%%%%%%%%%%%%%%%%%%%%%%%%%%%%%%%%%%%%%%%%
Using anticommuting Grassmann variables for calculating physical quantities for fermionic systems is
a well established technique. This concerns, in particular, the calculation of
expectation values of quasifree (Gaussian) fermion states. The idea is to
replace the linear combinations of canonically anticommuting fermi field operators with complex
coefficients by linear
combinations with coefficients that are anticommuting Grassmann numbers. As a consequence, these
linear combinations fulfill ``canonical commutation relations''. By interpreting tuples of
Grassmann numbers as ``phase space vectors'', a similar analysis can be carried out as it is known
for the bosonic case. In their article \cite{CahGlau99}, Cahill and Glauber used this technique to
analyze density operators for
fermionic states. These calculations are presented on a symbolic level by starting from a set of
computational rules with less focus on the underlying mathematical structure. The aim of this
paper is to close this gap and to provide the appropriate mathematical framework which
imports the Grassmann calculus into the description of fermion systems. We set up here a formalism
which can be viewed as the fermionic analog of a ``quantum
harmonic analysis on phase space'' (see \cite{WeQMPS}) and we provide here a collection of basic
propositions and theorems which help to simplify calculations. In addition to that, our approach
allows, up to certain extend, to consider also infinite
dimensional systems.

A direct way of explaining the GAR algebra for the finite dimensional case is given in terms of
standard creation and annihilation operators of fermionic modes: We consider $2n+k$ fermionic modes
and take all the creation operators $c_1^*,c_2^*,\cdots,c_{2n}^*,c_{2n+1}^*,\cdots , c_{2n+k}^*$ and
the last $k$ annihilation operators $c_{2n+1},\cdots c_{2n+k}$. Then we build the
subalgebra generated by these operators, being represented on the antisymmetic Fock space over
$\7C^{2n+k}$. Obviously, this algebra is not closed under the usual adjoint since we only take the
creation operators for the first $2n$ modes. In fact,
the first $2n$ annihilation operators generate the exterior algebra or {\em Grassmann algebra} over
$\7C^{2n}$. The last $k$ modes are identified with the usual Fermion algebra. In
order to implement a {\em complex conjugation} for Grassmann variables, we introduce a ``new''
adjoint $\ADJ$ on the GAR algebra. For a generator $c_i$ that belongs to the first $2n$ modes, it
is defined by $c_i^\ADJ:=c_{2n-i+1}$. For a generator that belongs to the last $k$ modes we just
take $c_j^\ADJ:=c_j^*$ the usual adjoint. It follows directly from this construction that the
defining
representation on antisymmetric Fock space is not a *-representation, i.e. it does not preserve the
adjoint. 

The GAR algebra consits of a {\em fermionic part}, that is generated by fermi field operators that
are linear combinations of creation and annihilation operators 
\begin{equation}
\label{eq:self-dual-fermi-fields}
B(f)=\sum_{i=2n+1}^{2n+k} f_+^i c_i^*+ f_-^i c_i 
\end{equation}
that fulfill the anti-commutation relations
\begin{equation}
\ACO{B(f)}{B(h)}=\sum_{i=2n+1}^{2n+k} (f^i_+h^i_-+ f^i_-h^i_+)\I \; .
\end{equation}
The fermionic part of the GAR algebra always corresponds to the underlying fermion system one
wishes to investigate. 

The {\em Grassmann part} of the GAR algebra is generated by the first $2n$ modes
$c_1^*,\cdots,c_{2n}^*$. Obviously, the GAR algebra possesses a natural $\7Z_2$-grading by looking
at the subspaces of {\em even} and {\em odd} operators. Here, we call a GAR operator to be even
(odd)
if it is a complex linear combination of even (odd) products of operators
$c_1^*,c_2^*,\cdots,c_{2n}^*,c_{2n+1}^*,\cdots , c_{2n+k}^*, c_{2n+1},\cdots c_{2n+k}$.

One of the main ideas behind introducing the GAR algebra is to build up an appropriate extension of
the fermion algebra in which the anti-commutation relations can be written in terms of
commutation relations by ``substituting'' the complex linear combinations of creation and
annihilation operators by linear combinations with ``anti-commuting'' variabes as coefficients.
These anti commuting variables, called {\em
Grassmann variables}, are linear
combinations of odd products of the creation operators $c_1^*,c_2^*,\cdots,c_{2n}^*$.
Suppose that $\xi=(\xi^1_+,\cdots\xi^k_+,\xi^1_-,\cdots\xi^k_-)$ is a vector of $2k$ odd
operators from the Grassmann part.
Then the linear combination 
\begin{equation}
\Phi(\xi)=\sum_{i=1}^{k} \xi_+^i c_{i+2n}^*+ \xi_-^i c_{i+2n} 
\end{equation}
is a well-defined operator inside the GAR algebra. These operators are the {\em Grassmann-Bose
fields}. A straight forward calculation shows that the commutation relation
\begin{equation}
\CO{\Phi(\xi)}{\Phi(\eta)}=-\sum_{i=1}^{k} \xi^i_+\eta^i_--\eta^i_+\xi^i_-=\sigma(\xi,\eta)
\end{equation}
are fulfilled. The bilinear form $\sigma$ can be interpreted a symplectic form with values in the
Grassmann algebra. Therefore, the Grassmann-Bose field operators $\Phi(\xi)$ fulfill a
kind of ``canonical commutation relations''. Here, one has to be aware of the fact
that the right hand side is not a complex multiple of the identity but an operator that belongs to
the
center of the GAR algebra. The {\em Grassmann-Weyl operators} are given by the
exponential
$\1w(\xi)=\exp(\Phi(\xi))$ and they fulfill the Weyl relations 
\begin{equation}
\1w(\xi+\eta)=\8e^{\frac{1}{2}\sigma(\xi,\eta)}\1w(\xi)\1w(\eta) \; .
\end{equation}

Our fermionic analog of a ``quantum harmonic analysis on phase space'' (see \cite{WeQMPS}) is based
on these structures.

Our paper is outlined into tree main sections, where the first section is for presenting basic
definitions and main results as precise as possible one one hand, and as less technical as possible
on the other hand. The two following sections are more technical in order to explain the mathematics
we are using. In Section~\ref{sec:basic}, we introduce all mathematical
concepts that are needed for our analysis. To make this part more readable, we postpone here the
discussion of technical details. We also present here the main results concerning an harmonic
analysis on fermionic phase space and we present its application to calculate the fidelity
between quasifree fermion states. Section~\ref{sec:math_struc} is dedicated to a detailed
discussion on the mathematical structure of the GAR algebra. Statements which have been claimed in
the previous section are proven here. Finally, Section~\ref{sec:grasscalc}, we review the concept
of Grassmann integration adopted to our analysis.

%%%%%%%%%%%%%%%%%%%%%%%%%%%%%%%%%%%%%%%%%%%%%%%%%%%%%%%%%%%%%%%%%%%%%%%%%%%%%%%%%%%%%%%%%%%%%%%%%%%%
\section{Basic definitions and main results}
%%%%%%%%%%%%%%%%%%%%%%%%%%%%%%%%%%%%%%%%%%%%%%%%%%%%%%%%%%%%%%%%%%%%%%%%%%%%%%%%%%%%%%%%%%%%%%%%%%%%
\label{sec:basic}
%%%%%%%%%%%%%%%%%%%%%%%%%%%%%%%%%%%%%%%%%%%%%%%%%%%%%%%%%%%%%%%%%%%%%%%%%%%%%%%%%%%%%%%%%%%%%%%%%%%%

%%%%%%%%%%%%%%%%%%%%%%%%%%%%%%%%%%%%%%%%%%%%%%%%%%%%%%%%%%%%%%%%%%%%%%%%%%%%%%%%%%%%%%%%%%%%%%%%%%%%
\subsection{Preliminaries}
\label{subsec:prelim}
%%%%%%%%%%%%%%%%%%%%%%%%%%%%%%%%%%%%%%%%%%%%%%%%%%%%%%%%%%%%%%%%%%%%%%%%%%%%%%%%%%%%%%%%%%%%%%%%%%%%
There are two basic notions that will play a essential role within the paper. Firstly, the concept
of a Banach *-algebra and the concept of a C*-algebra. We briefly recall these concepts here in
order to give a precise formulation of the mathematical structure we are going to use.
\begin{itemize}
\item 
A {\em Banach algebra} is a associative algebra and a Banach space with a norm $\|\cdot\|$, such
that the relation $\|AB\|\leq\|A\|\|B\|$ is fulfilled. 
\item
{\em An adjoint} is a continuous anti-linear involution $^\ADJ$ which fulfills
$(cA)^\ADJ=\bar c A^\ADJ$ for an operator $A$ and a complex number $c$ and the order in a product is
reversed $(AB)^\ADJ=B^\ADJ A^\ADJ$. 
\item 
A {\em Banach *-algebra} is a Banach algebra with an adjoint.
\item
A C*-algebra is a Banach *-algebra, where the adjoint $^*$ fulfills the C*-condition
$\|A^*A\|=\|A\|^2$. Vice versa, an adjoint $^*$ on a Banach algebra is called a {\em C*-adjoint} if
the C*-condition holds. 
\end{itemize}

In view of treating superspaces and supersymmetry Banach *-algebras has been used for instance in  
\cite{JadPil81}.

Later on it will be an important issue to distinguish the adjoint of a
C*-algebra with the adjoint of a generic *-algebra. The convention we will use here is to denote the
C*-adjoint of an operator $A$ by $A^*$
and the adjoint in a generic *-algebra by $A^\ADJ$.

In order to introduce the GAR algebra in a most general manner, we briefly recall here
{\em Araki's selfdual CAR algebra} \cite{MR0295702}. Let $H$ be a separable Hilbert space
with a complex conjugation $J$. Note that $H$ may be infinite dimensional. Then there exists a
unique C*-algebra $\CAR{H}{J}$, the so called {\em CAR algebra}, that is
generated by operators $B(f)$ with $f\in H$ such that $f\mapsto B(f)$ is a complex linear map, the
C*-adjoint of $B(f)$ is given by $B(f)^*=B(Jf)$ and the anti-commutator fulfills the relation
$\ACO{B(f)}{B(h)}=\SPROD{Jf}{h}\I$. 

\begin{rem}\em
The main advantage of using Araki's selfdual description of the CAR algebra  is, that it is
independent of the chosen representation. To obtain a representation in terms of creation and
annihilation operators, Equation~(\ref{eq:self-dual-fermi-fields}) is an example for the finite
dimensional Hilbert space $H=\7C^{2k}$, where the complex conjugation is given by
$J(f_-^1,\cdots,f_-^k,f_+^1,\cdots,f_+^k)=(\bar f_+^1,\cdots,\bar f_+^k,\bar
f_-^1,\cdots,\bar f_-^k)$.
\end{rem}

%%%%%%%%%%%%%%%%%%%%%%%%%%%%%%%%%%%%%%%%%%%%%%%%%%%%%%%%%%%%%%%%%%%%%%%%%%%%%%%%%%%%%%%%%%%%%%%%%%%%
\subsection{The GAR algebra}
%%%%%%%%%%%%%%%%%%%%%%%%%%%%%%%%%%%%%%%%%%%%%%%%%%%%%%%%%%%%%%%%%%%%%%%%%%%%%%%%%%%%%%%%%%%%%%%%%%%%
\label{subsec:GAR}
%%%%%%%%%%%%%%%%%%%%%%%%%%%%%%%%%%%%%%%%%%%%%%%%%%%%%%%%%%%%%%%%%%%%%%%%%%%%%%%%%%%%%%%%%%%%%%%%%%%%
The basic ingredients for construction the GAR algebra are a separable Hilbert space $H$, a complex
conjugation $J$ and a projection $Q$ that commutes with $J$. We consider the Hilbert space
$H_Q:=H\oplus Q^\perp H$ with the complex conjugation $J_Q:= R_Q(J\oplus Q^\perp J)$, where
$Q^\perp=\I-Q$ is the projection onto the orthogonal complement of $QH$ and the reflection $R_Q$ is
defined according to $R_Q(f\oplus h)=(Qf+h)\oplus Q^\perp f$ for $f\in H$, $h\in
Q^\perp H$. We associate to the triple $(H,Q,J)$ the CAR algebra
$\ENV{H}{J}{Q}:=\CAR{H_Q}{J_Q}$ in view of the following definition:

%%%%%%%%%%%%%%%%%%%%%%%%%%%%%%%%%%%%%%%%%%%%%%%%%%%%%%%%%%%%%%%%%%%%%%%%%%%%%%%%%%%%%%%%%%%%%%%%%%%%
\begin{defi}\em
\begin{enumerate}
\item
The {\em GAR algebra} $\GAR{H}{J}{Q}$ associated with the triple $(H,Q,J)$ is the norm-closed
subalgebra of the CAR algebra $\ENV{H}{J}{Q}$ that is generated by the {\em Grassmann-Fermi field
operators} $G(f):=B(f\oplus 0)$, $f\in H$.

\item
The {\em open core} $\PGAR{H}{J}{Q}$ of the GAR algebra is the subalgebra that consists of finite
sums of finite products of Grassmann-Fermi field operators.

\item
The CAR algebra $\ENV{H}{J}{Q}$ is called the {\em enveloping CAR algebra} of $\GAR{H}{J}{Q}$. 

\item
The {\em fermionic part} of the GAR algebra is the norm-closed subalgebra of $\ENV{H}{J}{Q}$ that
is generated by the Grassmann-Fermi fields $G(Qf)$, $f\in QH$. 

\item
The {\em Grassmann part} of the GAR algebra is the norm-closed subalgebra of $\ENV{H}{J}{Q}$ that 
is generated by the Grassmann-Fermi fields $G(Q^\perp f)$, $f\in Q^\perp H$. 
\end{enumerate}
\end{defi}
%%%%%%%%%%%%%%%%%%%%%%%%%%%%%%%%%%%%%%%%%%%%%%%%%%%%%%%%%%%%%%%%%%%%%%%%%%%%%%%%%%%%%%%%%%%%%%%%%%%%

As a consequence of this definition, the GAR algebra $\GAR{H}{J}{Q}$ is a Banach algebra and the
Grassmann-Fermi field operators fulfill
the anti-commutation relations
\begin{equation}
\ACO{G(f)}{G(h)}=\SPROD{Jf}{Qh}\I \\
\end{equation}
for $f,h\in H$. This can be verified by calculating the anti-commutator
$\ACO{G(f)}{G(h)}=\ACO{B(f\oplus 0)}{B(h\oplus
0)}=\SPROD{J_Q(f\oplus 0)}{h\oplus 0}\I=\SPROD{R_Q(Jf\oplus 0)}{h\oplus
0}\I=\SPROD{JQf}{h}\I=\SPROD{Jf}{Qh}\I$. From this calculation, it also follows that 
the fermionic part of $\GAR{H}{J}{Q}$ coincides with the CAR subalgebra
$\CAR{QH}{QJ}\subset\ENV{H}{J}{Q}$.

As already mentioned, inside the GAR algebra we can build linear combinations of fermion operators
with coefficients in the Grassmann algebra which yields the possibility of building fields
with ``canonical commutation relations'' inside the GAR algebra. For this purpose, it is also
important to have the concept of an adjoint. The problem is here, that the C*-adjoint in
the enveloping CAR algebra can not be used, since the GAR algebra is not not closed unter this
operation. Namely, for a generator $G(Q^\perp f)$ the C*-adjoint is given by $G(Q^\perp
f)^*=B(Q^\perp f\oplus 0)^*=B(0\oplus JQ^\perp f)\notin
\GAR{H}{J}{Q}$. Only the fermionic part is stable under the C*-adjoint. We shall see
(Proposition~\ref{prop:adj}) that there exists an adjoint
$^\ADJ\mathpunct:\GAR{H}{J}{Q}\to\GAR{H}{J}{Q}$ such that the GAR algebra
becomes a Banach *-algebra and that the adjoint $^\ADJ$ coincides with the C*-adjoint on the
fermionic part. Moreover, the adjoint $^\ADJ$ is uniquely determined by the relation
$G(f)^\ADJ=G(Jf)$ for $f\in H$.

%%%%%%%%%%%%%%%%%%%%%%%%%%%%%%%%%%%%%%%%%%%%%%%%%%%%%%%%%%%%%%%%%%%%%%%%%%%%%%%%%%%%%%%%%%%%%%%%%%%%
\begin{rem}\em
The Grassmann algebra can be regarded as a special case of the GAR algebra, where the projection
$Q$ is chosen to be zero. To be more precise, the {\em Grassmann algebra} $\GR{H}{J}$ over the pair
$H,J$ is defined as the GAR-algebra $\GR{H}{J}:=\GAR{H}{J}{0}$. On the other hand, the Grassmann
algebra can be constructed from the anti-symmetric tensor algebra
$\ATEN{H}{J}:=\bigoplus_{k\in\7N}\wedge^k H$ over the Hilbert space $H$. As a linear space, the
anti-symmetric tensor algebra $\ATEN{H}{J}$ is a dense subspace of the
anti-symmetric Fock space $\AFOCK{H}$ over $H$. Thus, we can equip $\ATEN{H}{J}$ with a scalar
product $\SPROD{}{}$. This gives rise to a further norm on $\ATEN{H}{J}$ that is given by 
$\HNO{\lambda}:=\sqrt{\SPROD{\lambda}{\lambda}}$. As we will see later, this norm is
{\em not} a Banach algebra norm, but it is continuous with respect to the Banach algebra norm
$\|\cdot\|$, i.e. $\HNO{A}\leq \|A\|$.
\end{rem}
%%%%%%%%%%%%%%%%%%%%%%%%%%%%%%%%%%%%%%%%%%%%%%%%%%%%%%%%%%%%%%%%%%%%%%%%%%%%%%%%%%%%%%%%%%%%%%%%%%%%

%%%%%%%%%%%%%%%%%%%%%%%%%%%%%%%%%%%%%%%%%%%%%%%%%%%%%%%%%%%%%%%%%%%%%%%%%%%%%%%%%%%%%%%%%%%%%%%%%%%%
\begin{rem}\em
By construction, the GAR algebra is isomorphic to the twisted (graded) tensor product (see
\cite{ChoBruDeWitt82,ChoBruDeWitt89} for this notion) of the fermionic part and the Grassmann part.
Following the analysis of \cite{JadPil83}, the Grassmann algebra can be regarded as the classical
limit of a field of CAR algebras. Analogously, the GAR algebra can be viewed as a {\em partial}
classical limit of a field of CAR algebras. The basic idea behind their work is to introduce for a
Hilbert space $H$, a complex conjugation $J$ and a positive number $\hbar>0$ the modified CAR
algebra $\CAR{H_\hbar}{J}$ where $H_\hbar$ is the Hilbert space with the scaled scalar product
$\SPROD{f}{h}_\hbar=\hbar\cdot\SPROD{f}{h}$. Roughly, in the classical limit of the field of CAR
algebras $(\CAR{H_\hbar}{J},\hbar>0)$ becomes the Grassmann algebra $\GR{H}{J}$ which is based on
the behavior anticommuator relations $\lim_{\hbar\to 0}\ACO{B_\hbar (f)}{B_\hbar
(h)}=\lim_{\hbar\to 0}\hbar\SPROD{f}{h}\I =0$. 

To view the GAR algebra as a partial
classical limit we consider the Hilbert space $H_{\hbar,Q}$ with the partially scaled scalar product
$\SPROD{f}{h}_\hbar:=\SPROD{f}{(Q+\hbar Q^\perp)h}$. The CAR algebra $\CAR{H_{\hbar,Q}}{J}$ is now
isomorphic to the twisted tensor product of $\CAR{QH}{QJ}$ and $\CAR{(Q^\perp H)_\hbar}{Q^\perp J}$.
Keeping in mind that the classical limit of the field of CAR algebras $(\CAR{(Q^\perp
H)_\hbar}{Q^\perp J}, \hbar>0)$ is the Grassmann algebra $\GR{Q^\perp H}{Q^\perp J}$, the GAR
algebra is the partial classical limit of the CAR algebra $\CAR{H_{\hbar,
Q}}{J}\to\GAR{H}{J}{Q}$.
\end{rem}
%%%%%%%%%%%%%%%%%%%%%%%%%%%%%%%%%%%%%%%%%%%%%%%%%%%%%%%%%%%%%%%%%%%%%%%%%%%%%%%%%%%%%%%%%%%%%%%%%%%%

%%%%%%%%%%%%%%%%%%%%%%%%%%%%%%%%%%%%%%%%%%%%%%%%%%%%%%%%%%%%%%%%%%%%%%%%%%%%%%%%%%%%%%%%%%%%%%%%%%%%
\begin{rem}\em
For the case that the projection $Q$ is a projection of even and finite dimension $2n$, there is a
further simple characterization of the GAR algebra. Namely, the GAR algebra can also be seen as a
matrix algebra with Grassmann valued entries. To verify this, we use the fact (as in the previous
remark) that the GAR algebra $\GAR{H}{J}{Q}$ is the twisted tensor product of the CAR algebra
$\CAR{QH}{QJ}$ and the Grassmann algebra $\GR{Q^\perp H}{Q^\perp J}$. Recall that the isomorphism
is given by $G(f)\mapsto G(Qf)\otimes \I+\Theta\otimes G(Q^\perp f)$, where $\Theta$ is the
reflection fulfilling $\Theta G(Qf)=-G(Qf)\Theta$. Moreover, the fermionic part
is isomorphic to the algebra $\8M_{2^{2n}}(\7C)$ of complex $2^{2n}\times 2^{2n}$ matrices. By
choosing a matrix basis $E_{ij}$, $i,j=1,\cdots 2^{2n}$, each operator in the GAR algebra can
uniquely be expanded as $A=\sum_{ij} E_{ij} \ A_{ij}$ where the operators $A_{ij}$ belong to the
Grassmann part. Thus the desired isomorphism identifies the operator $A$ with the matrix $(A_{ij})$
belonging to the algebra $\8M_{2^{2n}}(\GR{Q^\perp H}{Q^\perp J})$ of $2^{2n}\times 2^{2n}$
matrices with entries in the Grassmann part.
\end{rem}
%%%%%%%%%%%%%%%%%%%%%%%%%%%%%%%%%%%%%%%%%%%%%%%%%%%%%%%%%%%%%%%%%%%%%%%%%%%%%%%%%%%%%%%%%%%%%%%%%%%%

%%%%%%%%%%%%%%%%%%%%%%%%%%%%%%%%%%%%%%%%%%%%%%%%%%%%%%%%%%%%%%%%%%%%%%%%%%%%%%%%%%%%%%%%%%%%%%%%%%%%
\subsection{States}
%%%%%%%%%%%%%%%%%%%%%%%%%%%%%%%%%%%%%%%%%%%%%%%%%%%%%%%%%%%%%%%%%%%%%%%%%%%%%%%%%%%%%%%%%%%%%%%%%%%%
\label{subsec:positivity}
%%%%%%%%%%%%%%%%%%%%%%%%%%%%%%%%%%%%%%%%%%%%%%%%%%%%%%%%%%%%%%%%%%%%%%%%%%%%%%%%%%%%%%%%%%%%%%%%%%%%
The GAR algebra possesses a natural convex cone of positive elements. Firstly, the set of
positive linear functionals consits of all linear functionals
$\omega\mathpunct:\GAR{H}{J}{Q}\to\7C$ with $\omega(A^\ADJ A)\geq 0$. Secondly, the positive cone
$\GAR{H}{J}{Q}_+$ consits of all operators that have positive expectation values for all positive
functionals.

In order to analyze the positivity of operators, we introduce the norm closed two-sided ideal
$\IDL{H}{J}{Q}$ that is genrated by the {\em selfadjoint nilpotent operators} in the
Grassmann part $\GR{Q^\perp H}{Q^\perp J}$. Recall that an operator $Z$ is nilpotent if there exists
$n\in\7N$ with $Z^n=0$. It can be shown (see Proposition~\ref{prop:positivity}) that to each
positive functional $\omega$ on the GAR algebra $\GAR{H}{J}{Q}$ there exists a unique
positive functional $\omega'$ on the fermionic part $\CAR{QH}{QJ}$ such that 
\begin{equation}
\omega(A+Z)=\omega'(A)
\end{equation}
where $A$ is an operator in the
fermionic part $\CAR{QH}{QJ}$ and $Z$ belongs to the ideal $\IDL{H}{J}{Q}$. We refere the reader to
Subsection~\ref{subsec:ideals,positivity} for a more detailed
discussion. This shows that the positive functional on the GAR algebra are in one to one
correspondence with
the positive functionals on the fermionic part. 

Instead of considering complex valued 
functionals, the appropriate concept, as it turns out later,
to consider functionals from the GAR algebra into its Grassmann part. The GAR
algebra is equipped with a natural right module structure over the Grassmann part via multiplication
from the right. For our purpose, the appropriate method is to extend a state on the fermionic part
as a right module homomorphism. For a linear functional $\omega$ on the CAR algebra $\CAR{QH}{QJ}$,
we are seeking for a linear map $\bS\omega\mathpunct:\GAR{H}{J}{Q}\to\GR{Q^\perp H}{Q^\perp J}$
which fulfills the condition
\begin{equation}
\label{eq:ext}
\bS\omega(A\lambda)=\omega(A)\lambda \; .
\end{equation}
We call $\bS\omega$ the {\em G-extension} of $\omega$ to the GAR algebra. To obtain the
G-extension of a state on the fermionic part, we use the fact that the eveloping CAR algebra
$\ENV{H}{J}{Q}$ can be identified with the twisted (graded) tensor product of the fermionic part and
the enveloping CAR algebra of the Grassmann part:
\begin{equation}
\ENV{H}{J}{Q}=\CAR{QH}{QJ}\tilde\otimes\9E
\end{equation}
where $\9E$ denotes the enveloping CAR algebra of the Grassman part. For a vector
$f\oplus h\in H\oplus Q^\perp H$ the corresponding fermi field operator is identified with the
tensor product by
\begin{equation}
B(f\oplus h )=B(Qf)\otimes\I +\Theta_Q\otimes B(Q^\perp f\oplus h)
\end{equation}
where $\Theta_Q$ is the relection that implements the parity automorphism $\Theta_Q
B(Qf)\Theta_Q=-B(Qf)$. As a Banach space , the enveloping CAR algebra is identified with a tensor
product of two C*-algebras. With respect to the positivity structure of this tensor product, the
linear map $\omega\otimes \8{id}_\9E$ is completely positive. Therefore, it is bounded as a map
between Banach spaces. The G-extension is now given by the restriction to the GAR algebra
\begin{equation}
\bS\omega:=\omega\otimes\8{id}_\9E|_{\GAR{H}{J}{Q}}
\end{equation}
which becomes a bounded map from the GAR algebra into its Grassmann part. By construction, the 
condition (\ref{eq:ext}) is fulfilled.

\begin{rem}\em
For our further analysis, the essential property of the G-extension is to be a right module
homomorphism (\ref{eq:ext}), whereas positivity with respect to the adjoint of the GAR is not
essential. Anyway, form the above construction, we cannot conclude directly that the G-extension is
positive since the GAR algebra has a different positivity structure than the tensor product
$\CAR{QH}{QJ}\otimes\9E$. 
\end{rem}

%%%%%%%%%%%%%%%%%%%%%%%%%%%%%%%%%%%%%%%%%%%%%%%%%%%%%%%%%%%%%%%%%%%%%%%%%%%%%%%%%%%%%%%%%%%%%%%%%%%%
\subsection{Anticommutative phase space}
%%%%%%%%%%%%%%%%%%%%%%%%%%%%%%%%%%%%%%%%%%%%%%%%%%%%%%%%%%%%%%%%%%%%%%%%%%%%%%%%%%%%%%%%%%%%%%%%%%%%
\label{subsec:antisym-phase-space}
%%%%%%%%%%%%%%%%%%%%%%%%%%%%%%%%%%%%%%%%%%%%%%%%%%%%%%%%%%%%%%%%%%%%%%%%%%%%%%%%%%%%%%%%%%%%%%%%%%%%
We have introduced the GAR algebra in terms of the Grassmann-Fermi field operators $G(f)$. In this
section we introduce a different family of field operators, called Grassmann-Bose fields, that also
generate, together with the unit operator, the GAR algebra. It turns out that these fields 
fulfill a graded version of the canonical commutation relations. 

We introduce the {\em anticommutative phase space} as the tensor
product $\PHSP{H}{J}{Q}:=QH\otimes\GR{Q^\perp H}{Q^\perp J}$. We are considering here a tensor
product of a Hilbert space and a Grassmann algebra closed with respect to
the projective cross norm $\PCN{\cdot}$ (see for instance \cite{Tak02}). Since the projective cross
norm is the largest among all
cross norms, it follows that the anticommutative phase space $\PHSP{H}{J}{Q}$
can be identified with a linear subspace of
the Grassmann algebra $\GR{H}{J}$ by the continuous embedding which identifies the tensor
product $f\otimes\lambda$ with the operator $\Lambda(f)\lambda$. In order to express the
(anti)commutation relations for the Grassmann-Bose fields, we equip the anticommutative phase space
$\PHSP{H}{J}{Q}$ with a continuous Grassmann valued inner product. This {\em rigging
map} $\RIGG{\cdot}{\cdot}{Q}$ is determined on pure tensor products by
$\RIGG{f\otimes\lambda}{h\otimes\mu}{Q}:=\SPROD{f}{h}\lambda^\star\mu$. 

%%%%%%%%%%%%%%%%%%%%%%%%%%%%%%%%%%%%%%%%%%%%%%%%%%%%%%%%%%%%%%%%%%%%%%%%%%%%%%%%%%%%%%%%%%%%%%%%%%%%
\begin{rem}\em
If $Q$ is a projection of finite and even rank $2n$, then the anticommutative phase space is simply
isomorphic to the $2n$-fold cartesian product of the Grassmann algebra. This can be seen by
choosing a real orthonormal basis $(e^i)_{i=1,\cdots,2n}$ of $QH$. Each phase space vector $\xi$
can be uniquely expanded as $\xi=\sum_i e^i\otimes \xi_i$, where $\xi_i$ is an operator from the
Grassmann part. Thus an isomorphism between $\PHSP{H}{J}{Q}$ and $\GR{Q^\perp H}{Q^\perp J}^{2n}$
is given by $\xi\mapsto(\xi_1,\cdots,\xi_{2n})$.  
\end{rem}
%%%%%%%%%%%%%%%%%%%%%%%%%%%%%%%%%%%%%%%%%%%%%%%%%%%%%%%%%%%%%%%%%%%%%%%%%%%%%%%%%%%%%%%%%%%%%%%%%%%%

The Grassmann-Bose field $\Phi$ is a right module homomorphism that associates to each
phase space vector $\xi$ an operator $\Phi(\xi)$ in the GAR algebra. This map is determined on pure
tensor products $\xi=f\otimes\lambda$ according to 
\begin{equation}
\Phi(f\otimes\lambda):=G(f)\lambda
\end{equation}
where $f\in QH$ and $\lambda$ belonging to the Grassmann part. Note that the inequality
$\|\Phi(\xi)\|\leq \PCN{\xi}$ holds which implies that the map $\Phi$ is continuous and can uniquely
be extended to the full anticommutative phase space. 

The $\7Z_2$-grading of the Grassmann part induces a direct sum decomposition of the
Banach space $\PHSP{H}{J}{Q}=\PHSP{H}{J}{Q}_0\oplus\PHSP{H}{J}{Q}_1$ with
$\PHSP{H}{J}{Q}_q:=QH\otimes\GR{Q^\perp H}{Q^\perp J}_{q+1}$. Moreover, we introduce a complex
conjugation according to $(f\otimes\lambda)^\ADJ:=(-1)^{q+1} Jf\otimes\lambda^\ADJ$ with
$\lambda\in\GR{Q^\perp H}{Q^\perp J}_q$. The grading and the complex conjugation are compatible
with the identification of $\PHSP{H}{J}{Q}$ as a closed linear subspace of the Grassmann algebra
$\GR{H}{J}$. Here the tensor product $f\otimes\lambda$ is just identified with the operator
$\Lambda(f)\lambda$. The complex conjugation and the grading in $\PHSP{H}{J}{Q}$ are nothing else
but the adjoint and the grading within the ambient Grassmann algebra. The GAR relations can be
expressed in terms of graded commutators. Recall that the graded commutator of $\GCO{A}{B}$ is
given by the commutator $\CO{A}{B}$ if $A$ or $B$ are even, and by the anticommutator $\ACO{A}{B}$
if both $A$ and $B$ are odd. 

It follows from the construction of the Grassmann-Bose fields that for a pair of phase space vectors
$\xi,\eta\in\PHSP{H}{J}{Q}$, the graded commutator fulfills
\begin{equation}
\GCO{\Phi(\xi)}{\Phi(\eta)}=\RIGG{\xi^\ADJ}{\eta}{Q}  \; .
\end{equation}
Moreover, the adjoint fulfills the identity $\Phi(\xi)^\ADJ=\Phi(\xi^\ADJ)$. For the
particular case, that $\xi,\eta$ are even elements in $\PHSP{H}{J}{Q}$ the Grassmann-Bose
field fulfill the canonical commutation relations:
\begin{equation}
\CO{\Phi(\xi)}{\Phi(\eta)}=\RIGG{\xi^\ADJ}{\eta}{Q} \; .
\end{equation}
Note that the restriction of the rigging map to the even subspace is antisymmetric, i.e.
$\RIGG{\xi^\ADJ}{\eta}{Q}=-\RIGG{\eta^\ADJ}{\xi}{Q}$ for
$\xi,\eta$ even. As for the usual formulation of the canonical commutation relation, the commutator
belongs to the center of the GAR algebra. The main difference is here, that the Grassmann-Bose
field operators are bounded in norm. This is no contradiction, since we are dealing here with
Banach *-algebras (rather than C*-algebras).

For an even $\xi\in\PHSP{H}{J}{Q}_0$ the exponential $\1w(\xi):=\exp(\Phi(\xi))$ of the
Grassmann-Bose field operator is well defined. We call $\1w(\xi)$ the \textit{Grassmann-Weyl
operator} for $\xi$. Since field operators $\Phi(\xi),\Phi(\eta)$ are even, the Grassmann-Weyl
operators fulfill the relations
\begin{equation}
\1w(\xi+\eta)=\8e^{\frac{1}{2}\RIGG{\xi^\ADJ}{\eta}{Q}}\1w(\xi)\1w(\eta) \; .
\end{equation}
As for ordinary Weyl operators, the map $\xi\mapsto \1w(\xi)$ is a projective representation of the
additive group $\PHSP{H}{J}{Q}_0$ where the factor system belongs to the center of the GAR algebra. 

Obviously the Grassmann-Weyl operator $\1w(\xi)$ is unitary only if $\xi=\xi^\ADJ$ is selfadjoint.
If we restrict the Grassmann-Weyl system to selfadjoint phase space vectors, then the value of the
rigging map $\RIGG{\xi^\ADJ}{\eta}{Q}=\RIGG{\xi}{\eta}{Q}=-\RIGG{\xi}{\eta}{Q}^\ADJ$ is
anti-selfadjoint and the exponential in the Grassmann-Weyl relation is also unitary.
Note that $\PHSP{H}{J}{Q}_0$ should be seen as a ``complexified'' anticommutative phase space and
the Grassmann-Weyl operators are directly constructed as a kind of analytic continuation from the
real part.

%%%%%%%%%%%%%%%%%%%%%%%%%%%%%%%%%%%%%%%%%%%%%%%%%%%%%%%%%%%%%%%%%%%%%%%%%%%%%%%%%%%%%%%%%%%%%%%%%%%%
\subsection{Towards a harmonic analysis on anticommutative phase space}
%%%%%%%%%%%%%%%%%%%%%%%%%%%%%%%%%%%%%%%%%%%%%%%%%%%%%%%%%%%%%%%%%%%%%%%%%%%%%%%%%%%%%%%%%%%%%%%%%%%%
\label{subsec:harmonic-analysis}
%%%%%%%%%%%%%%%%%%%%%%%%%%%%%%%%%%%%%%%%%%%%%%%%%%%%%%%%%%%%%%%%%%%%%%%%%%%%%%%%%%%%%%%%%%%%%%%%%%%%
The concept of anticommutative phase space can be used to perform a kind of ``harmonic
analysis'' that is analogous to the analysis of the bosonic
case \cite{WeQMPS}. As for the case of ordinary
symplectic vector spaces \cite{WeQMPS}, we introduce here the analogous concept of {\em
convolution} and {\em Fourier transform}. This requires to ``integrate'' over
antisymmetric phase space. Here the Brezin-Grassmann integration (see for instance
\cite{VladVol84,VladVol84b,ChoBruDeWitt89}) turns out to be the appropriate notion which we recall
here. Before we continue our discussion, we mention the following:

In order to perform integration with respect to Grassmann variables we
have to consider the algebra of functions that can be integrated. These functions are appropriate
polynomials of Grassmann variables $\xi\in\PHSP{H}{J}{Q}_0$ with values in a right module over the
ring $\GR{Q^\perp H}{Q^\perp J}$. If the underlying ring structure is clear from the context we just
briefly say ``right module''. We assume here, that the projection $Q$ has finite even rank
$\8{dim}(Q)=2n$, which corresponds to an integration over a finite dimensional space.

We need to integrate functions with values in a right module $\9E$ that admit a {\em polynomial
representation} in terms of Grassmann variables. However, this
representation has some ambiguities which causes some problems in defining the Grassmann integral.
A polynomial representation is obtained from a real orthonormal basis $(e^i)_{i\in N}$ of $QH$
that is indexed by the ordered set $N=\{1,\cdots,\8{dim}(Q)\}$. Any vector $\xi\in\PHSP{H}{J}{Q}_0$
can be expanded in this basis as $\xi=\sum_i\Lambda^i\xi_i$ with $\xi_i\in\GR{Q^\perp H}{Q^\perp
J}_1$. With respect to this basis, the polynomial representation of a G-holomorphic function $F$ is
given by 
\begin{equation}
F(\xi)=\sum_{I\subset N} F^I \ \xi_I \, .
\end{equation}
Here the coefficients $F^I$ are contained in the right module $\9E$. The monomial $\xi_I$ which is
associated to an
ordered subset $I=\{i_1<i_2<\cdots<i_k\}\subset N$ is given by $\xi_I:=\xi_{i_1}\cdots\xi_{i_k}$.

For a given polynomial representation, the Brezin-Grassmann integral of $F$ over a form of highest
degree ($2n$-form) $v=v_N e^1\wedge\cdots\wedge e^{2n}$ in $\GR{QH}{QJ}$ is defined according to
\begin{equation}
\int_Q v(\xi) \ F(\xi)=v_{N} F^{N} \ , .
\end{equation}
The problem with this definition is that, in general, the relation 
\begin{equation}
\label{eq:zero}
\sum_{I\subset N} F^I \ \xi_I=0
\end{equation}
is non-trivial in the sense that (\ref{eq:zero}) holds for non-zero coefficients $F^I$. This means
that the polynomial representation for a given basis is not unique and the integral may not be
well defined. However, for some right modules $\9E$ the highest coefficient $F^N$ is unique. This
is the case if $\9E$ fulfills the following condition: If $X\lambda_1\cdots\lambda_{2n}=0$ for
all odd Grassmann operators $\lambda_1,\cdots,\lambda_{2n}$, then $X=0$. This property holds for
$\9E=\GAR{H}{J}{Q}$ and $\9E=\GR{Q^\perp H}{Q^\perp J}$, if the complementary projection $Q^\perp$
is infinite dimensional. A more detailed analysis of this issue is postponed to
Subsection~\ref{subsec:G-hol}. In what follows, we assume that $Q$ is a projection of finite
dimesion $2n$ and its complement $Q^\perp$ is infinite dimensional. 

We introduce the Banach space $\RGHOM{H}{J}{Q}$ of bounded right module homomorphisms from
the GAR algebra into its Grassmann part. As already mentioned, the space $\PHSP{H}{J}{Q}_0$ can be
regarded as anticommutative phase space with a complex structure that is given by the adjoint $\ADJ$
where the rigging map $\RIGG{}{}{Q}$ induces a Grassmann valued symplectic form on
$\PHSP{H}{J}{Q}_0$. The phase space translations act by automorphisms on the GAR algebra by the
adjoint action of the Grassmann-Weyl operators 
\begin{equation}
\alpha_\xi(A)=\1w(-\xi)A\1w(\xi) \; .
\end{equation}
Note that $\alpha_\xi$ is a *-automorphism, if and only if, the
operator $\xi$ is selfadjoint, i.e. it belongs to the real part of $\PHSP{H}{J}{Q}_0$.

By fixing a selfadjoint and normalized form $v$ of highest degree with respect to $Q$, we are now
prepared to define the following convolutions:
\begin{enumerate}
\item 
The convolution of two G-holomorphic functions
is a G-holomorphic function that is given by
\begin{gather}
\RHOL{H}{J}{Q}\times\RHOL{H}{J}{Q}\ni(f_1,f_2)\to f_1\CON f_2\in\RHOL{H}{J}{Q}
\nonumber
\\
(f_1\CON f_2)(\xi):=\int_Q v(\eta) \ f_1(\eta)f_2(\xi-\eta) \; .
\end{gather}

\item 
The convolution of an operator in the GAR algebra with a G-holomorphic
function is the operator in the GAR algebra that is given by
\begin{gather}
\GAR{H}{J}{Q}\times\RHOL{H}{J}{Q}\ni(A,f)\to A\CON f\in\GAR{H}{J}{Q}
\nonumber
\\
A\CON f:=\int_Q v(\eta) \ \alpha_{-\eta}(A) \ f(\eta) \; .
\end{gather}

\item
The convolution of a right module homomorphism with an operator in the GAR algebra is a
G-holomorphic function
that is given by
\begin{gather}
\RGHOM{H}{J}{Q}\times\GAR{H}{J}{Q}\ni(\bS\varphi,A)\to\bS\varphi\CON A\in\RHOL{H}{J}{Q}
\nonumber
\\
(\bS\varphi\CON A)(\xi):=\bS\varphi(\alpha_\xi(A)) \; .
\end{gather}
\end{enumerate}

As for the convolution, we also have three different cases for applying the {\em Fourier transform}.
\begin{enumerate}
\item 
The Fourier transform maps an operator of the reduced GAR algebra to a G-holomorphic function
according to:
\begin{gather}
\FOUR:\GAR{H}{J}{Q}\to\RHOL{H}{J}{Q}
\nonumber
\\
(\FOUR A)(\xi):={\1w}(-\xi)\int_Q v(\eta) \ \alpha_{-\eta}(A)
\ \8e^{\RIGG{\xi^\ADJ}{\eta}{Q}} \; .
\end{gather}
Up to now, the function $\FOUR A$ can have values in the GAR algebra, but
we will show that the range of the function is indeed fully contained in the Grassmann part.

\item 
The Fourier transform maps a G-holomorphic function to a G-holomorphic function according to:
\begin{gather}
\FOUR:\RHOL{H}{J}{Q}\to\RHOL{H}{J}{Q}
\nonumber
\\
(\FOUR f)(\xi):=\int_Q v(\eta) \ f(\eta) \ \8e^{\RIGG{\xi^\ADJ}{\eta}{Q}} \; .
\end{gather}

\item
Finally, the Fourier transform of a right module homomorphism is the G-holomorphic function which is
given by its expectation values of the Grassmann-Weyl operators in the GAR algebra:
\begin{gather}
\FOUR:\RGHOM{H}{J}{Q}\to \RHOL{H}{J}{Q}
\nonumber
\\
(\FOUR\bS\varphi)(\xi)={\bS\varphi}({\1w}(\xi)) \; .
\end{gather}
\end{enumerate}

%%%%%%%%%%%%%%%%%%%%%%%%%%%%%%%%%%%%%%%%%%%%%%%%%%%%%%%%%%%%%%%%%%%%%%%%%%%%%%%%%%%%%%%%%%%%%%%%%%%%
\subsection{Main Results}
%%%%%%%%%%%%%%%%%%%%%%%%%%%%%%%%%%%%%%%%%%%%%%%%%%%%%%%%%%%%%%%%%%%%%%%%%%%%%%%%%%%%%%%%%%%%%%%%%%%%
\label{subsec:main-results}
%%%%%%%%%%%%%%%%%%%%%%%%%%%%%%%%%%%%%%%%%%%%%%%%%%%%%%%%%%%%%%%%%%%%%%%%%%%%%%%%%%%%%%%%%%%%%%%%%%%%
A useful fact is that all operators of the GAR algebra $\GAR{H}{J}{Q}$ can be represented in terms
of Grassmann integrals as stated by the following theorem: 

%%%%%%%%%%%%%%%%%%%%%%%%%%%%%%%%%%%%%%%%%%%%%%%%%%%%%%%%%%%%%%%%%%%%%%%%%%%%%%%%%%%%%%%%%%%%%%%%%%%%
\begin{thm}
\label{thm:integral-rep}
The Fourier transform $\FOUR$ maps each operator $A\in\GAR{H}{J}{Q}$ of the GAR
algebra to a G-holomorphic function in $\RHOL{H}{J}{Q}$ such that the identity
\begin{equation}
\label{eq:pre-decomp}
A =\int_Q v(\xi) \  {\1w}(\xi) \ (\FOUR A)(\xi)
\end{equation}
holds.
\end{thm}
%%%%%%%%%%%%%%%%%%%%%%%%%%%%%%%%%%%%%%%%%%%%%%%%%%%%%%%%%%%%%%%%%%%%%%%%%%%%%%%%%%%%%%%%%%%%%%%%%%%%

The Fourier transform maps the convolution of objects into their product of
Fourier transforms as stated by the following theorem. 

%%%%%%%%%%%%%%%%%%%%%%%%%%%%%%%%%%%%%%%%%%%%%%%%%%%%%%%%%%%%%%%%%%%%%%%%%%%%%%%%%%%%%%%%%%%%%%%%%%%%
\begin{thm}
\label{thm:convol-four}
Let $f,f'$ be G-holomorphic functions, let $A$ an operator of the GAR algebra and let
$\bS\varphi$ be a right module homomorphism. Then the identities 
\begin{equation}
\label{eq:conv-four}
\begin{split}
&\FOUR(f\CON f')=\FOUR f \ \FOUR  f'
\\
&\FOUR(\bS\varphi\CON A)=\FOUR\bS\varphi \  \FOUR A
\\
&\FOUR(A\CON f)=\FOUR A  \ \FOUR f
\end{split}
\end{equation}
are valid.
\end{thm}
%%%%%%%%%%%%%%%%%%%%%%%%%%%%%%%%%%%%%%%%%%%%%%%%%%%%%%%%%%%%%%%%%%%%%%%%%%%%%%%%%%%%%%%%%%%%%%%%%%%%

The integral representation for operators in the GAR algebra can be used to calculate the
expectation values of a right module homomorphism on the GAR algebra in terms of the expectation
values of Grassmann-Weyl operators. According to the discussion
above,
the identities 
\begin{equation}
{\bS\omega}(A)=\int_Q v \ \FOUR\bS\omega \ \FOUR A =
\int_Q v \ \FOUR (\bS\omega\CON A)
\end{equation}
are valid.

%%%%%%%%%%%%%%%%%%%%%%%%%%%%%%%%%%%%%%%%%%%%%%%%%%%%%%%%%%%%%%%%%%%%%%%%%%%%%%%%%%%%%%%%%%%%%%%%%%%%
\begin{cor}
\label{cor:RN}
Let $\bS\omega,\bS\varphi$ bounded right module homomorphisms. If 
$\FOUR\bS\varphi$ is a divisor of $\FOUR\bS\omega$ within the ring of G-homomorphic functions, then
the Radon-Nikodym type relation
\begin{equation}
{\bS\omega}(A)=\int_Q v \ \frac{\FOUR\bS\omega}{\FOUR\bS\varphi} \ \FOUR(\bS\varphi\CON A)
\end{equation}
is valid for all operators $A$ of the GAR algebra.
\end{cor}
%%%%%%%%%%%%%%%%%%%%%%%%%%%%%%%%%%%%%%%%%%%%%%%%%%%%%%%%%%%%%%%%%%%%%%%%%%%%%%%%%%%%%%%%%%%%%%%%%%%%

For the presentation of our next results, we need to recall the definition of the {\em
quasifree fermion states} as well as the definition of the {\em Pfaffian} of a real antisymmetric
matrix, as well as how to calculate Gaussian Grassmann integrals.

\vspace{2em}

\noindent
{\bf Quasifree states:} 
Each quasifree state on the fermionic part $\CAR{QH}{QJ}$ is in one-to-one
correspondence with its covariance matrix. This is a linear operator
$S$ on $QH$ with $0<S\leq \I$ and it has to fulfill the constraint $S+JSJ=\I$.
The expectation values of the state $\omega_S$ are related to its
covariance matrix $S$ by the following condition on the two-point correlation
function:
\begin{equation}\label{eq:twopoint}
\omega_S( B(f) B(h))=\SPROD{J f}{Sh}
\end{equation}
with $f,h\in QH$.
All higher correlation functions can be expressed in terms of sums of products
of two-point functions according to Wick theorem, where only the expectation values of an even
product of Fermi field operators are non vanishing. It is well
known that a quasifree state is pure if and only if its covariance matrix $S=P$ is a projection,
called {\em basis projection}.

\vspace{2em}

\noindent
{\bf Pfaffian:}
It is well known, that Gaussian Grassmann integrals can be expressed in terms of the Pfaffian of
the corresponding covariance matrix. The $n$th power of a two-form $a\in\GR{QH}{QJ}$ is a $2n$-form
in
$\GR{QH}{QJ}$ and therefore proportional to any other $2n$-form. By fixing a selfadjoint normalized
$2n$-form $v$, i.e. $v=v^\ADJ$ and $\SPROD{v}{v}=1$, there
exists a complex number $\8{Pf}_{[v]}(a)$, called the {\em Pfaffian} of $a$ that is
uniquely determined by 
\begin{equation}
(n!)^{-1}a^n=(n!)^{-1}\SPROD{v}{a^n} v =\8{Pf}_{[v]}(a)v \; .
\end{equation}
Now, let $A$ be a linear operator on $QH$, then there exists
a unique two-form $a$ such that $\SPROD{a^\ADJ}{f\wedge h}=\SPROD{Jf}{Ah}$, where the two-form $a$
only depends on the {\em $J$-antisymmetric part} $(A-JA^*J)/2$ of the operator $A$. Note, that the
$J$-antisymmetry is related to {\em the transpose} $A\mapsto JA^*J$. The Pfaffian of
a $J$-antisymmetric operator $A$ is now defined as
\begin{equation}
\8{Pf}_{[v,J]}(A):=\8{Pf}_{[v]}(a) \, .
\end{equation}
where $a$ is the two-form fulfilling the identity $\SPROD{a^\ADJ}{f\wedge h}=\SPROD{Jf}{Ah}$. If
we restrict to the real subspace in $QH$ that is given by $Jf=f$, we obtain therefore the standard
definition of the Pfaffian for a real antisymmetric operator. In particular, the determinant of an
antisymmetric operator with respect to the $2n$-form $v$ is given by
$\det(A)v=\Gamma(A)v$, where $\Gamma(A)$ denotes the second quantized operator of $A$
on the antisymmetric Fock space over $QH$ that is given by $\Gamma(A)(f_1\wedge\cdots\wedge
f_k)=Af_1\wedge\cdots\wedge Af_k$. We recall here the well known identity 
\begin{equation}
\det(A)=\8{Pf}_{[v,J]}(A)^2 \; .
\end{equation}
Note that due to the condition $A=-JA^*J$ the left hand side implicitly also depends on the complex
conjugation $J$. However, whereas the determinant can be defined for any linear operator on $QH$,
the Pfaffian is only defined on the {\em real} linear subspace of $J$-antisymmetic operators.

%%%%%%%%%%%%%%%%%%%%%%%%%%%%%%%%%%%%%%%%%%%%%%%%%%%%%%%%%%%%%%%%%%%%%%%%%%%%%%%%%%%%%%%%%%%%%%%%%%%%
\begin{thm}
\label{thm:g-extension}
For each covariance operator $S$ on $QH$ the Fourier transform of the G-extension $\bS\omega_S$ of
the quasifree state $\omega_S$ is given by 
\begin{equation}
\label{eq:char}
\FOUR{\boldsymbol\omega}_S(\xi)=\8e^{-\frac{1}{2}\RIGG{\xi^\ADJ}{S\xi}{Q}} \; .
\end{equation}
Moreover, let $P$ be a basis projection and let $E_P$ be
the support projection of the pure quasifree state $\omega_P$. Then for a normalized selfadjoint
form $v$ of highest degree with respect to $Q$ the identity
\begin{equation}
\label{equ:four_conv}
\FOUR(\bS\omega_P\CON E_P)(\xi)=\epsilon_{[v,P]} \ \8e^{-\RIGG{\xi^\ADJ}{P\xi}{Q}}
\end{equation}
is valid, where the sign $\epsilon_{[v,P]}=\8{Pf}_{[v,J]}(\I-2P)=\pm 1$ depends on the 
orientation of the form $v$ and the reflection $\I-2P$.
\end{thm}
%%%%%%%%%%%%%%%%%%%%%%%%%%%%%%%%%%%%%%%%%%%%%%%%%%%%%%%%%%%%%%%%%%%%%%%%%%%%%%%%%%%%%%%%%%%%%%%%%%%%

This theorem states that in analogy to the bosonic expectation values of Weyl operators, the
expectation values for the displacement operators for G-extended quasifree states are also of
Gaussian type. In particular, the relation (\ref{equ:four_conv}) is derived by calculating a
Gaussian Grassmann integral with respect to the covariance $A=\I-2P$. Recall, that for a
$J$-antisymmetric operator $A$ on $QH$ and a $v$ normalized selfadjoint $2n$-form the corresponding 
Gaussian integral can be calculated according to
\begin{equation}
\label{equ:gaussian}
\int_Qv(\xi) \ \8e^{\frac{1}{2}
\RIGG{\xi^\ADJ}{A\xi}{Q}+\RIGG{\eta^\ADJ}{\xi}{Q}}=
\8{Pf}_{[v,J]}(A)\8e^{\frac{1}{2}\RIGG{\eta^\ADJ}{A^{-1}\eta}{Q}} \; .
\end{equation}
Inserting $A=\I-2P$ into the above identity (\ref{equ:gaussian}) yields the identity
(\ref{equ:four_conv}).

%%%%%%%%%%%%%%%%%%%%%%%%%%%%%%%%%%%%%%%%%%%%%%%%%%%%%%%%%%%%%%%%%%%%%%%%%%%%%%%%%%%%%%%%%%%%%%%%%%%%
\subsection{Applications: Calculating the fidelity of quasifree states}
%%%%%%%%%%%%%%%%%%%%%%%%%%%%%%%%%%%%%%%%%%%%%%%%%%%%%%%%%%%%%%%%%%%%%%%%%%%%%%%%%%%%%%%%%%%%%%%%%%%%
\label{subsec:applications}
%%%%%%%%%%%%%%%%%%%%%%%%%%%%%%%%%%%%%%%%%%%%%%%%%%%%%%%%%%%%%%%%%%%%%%%%%%%%%%%%%%%%%%%%%%%%%%%%%%%%
In this section, we give an explicite formula to calculate the fidelity between a pure quasifree
state and another arbitrary quasifree state. Let us recall the fidelity between two states
$\omega_1$ and $\omega_2$ on a general finite dimensional C*-algebra $\6A$. Let $L_2(\6A)$ be the
Hilbert space of Hilbert-Schmidt operators with respect to a faithful trace $\tr$ on $\6A$. A
Hilbert-Schmidt operator
$V\in L_2(\6A)$ implements a state $\omega$ if $\omega(A)=\SPROD{V}{A V}=\tr(V^*AV)=\tr(VV^*A)$
holds for all $A\in \6A$. In this case we write $v\in S(\omega)$. Clearly, for $V\in S(\omega)$ the
operator $VV^*$ is the density operator that corresponds to the state $\omega$. 

As long as we consider finite dimensions, all states can be implemented that way. For two states
$\omega,\varphi$, the fidelity is given by 
\begin{equation}
F(\omega,\varphi)=\sup_{V\in S(\omega),W\in S(\varphi)}|\SPROD{V}{W}| \; .
\end{equation}
As it has been shown by Bures \cite{Bur69} (compare also \cite{KretSchlWer08}), the fidelity is
related to the norm distance of states, independent
of the dimesion of the underlying algebra, by the inequality
\begin{equation}
\frac{1}{4}\|\omega-\varphi\|^2\leq 2(1-F(\omega,\varphi))\leq\|\omega-\varphi\| \; .
\end{equation}

The fidelity can simply be calculated if we choose one of the states $\omega$ to be pure. In this
case, for each implementing Hilbert-Schmidt operator $V\in S(\omega)$ there exists a unitary $U\in
\6A$ such that $V=EU$, where $E\in\6A$
is the unique rank-one projection (density operator) that corresponds to $\omega$, i.e.,
$\omega(A)=\tr(EA)$. Thus we observe for each $W\in S(\varphi)$ that
$\SPROD{EU}{W}^2=\tr(W^*EU)\tr(U^*E W)=\tr(UW^*E)\tr(E WU^*)=\tr(W^*EW)=\varphi(E)$ and we get for
the fidelity:
\begin{equation}
F(\omega,\varphi)^2=\varphi(E) \; . 
\end{equation}

As we have seen above, the GAR framework allows to calculate expectation values of operators in
terms of Grassmann integrals. In the case of quasifree states, we see from
Theorem~\ref{thm:g-extension} that we are faced here with calculating Gaussian Grassmann integrals
only. This leads to the following theorem:

%%%%%%%%%%%%%%%%%%%%%%%%%%%%%%%%%%%%%%%%%%%%%%%%%%%%%%%%%%%%%%%%%%%%%%%%%%%%%%%%%%%%%%%%%%%%%%%%%%%%
\begin{thm}
\label{thm:fidelity}
Let $S$ be two covariance operator and let $P$ be a basis projection on $QH$ and let
$\omega_{S},\omega_{P}$ be the corresponding quasifree states. Then the 
relation
\begin{equation}
F(\omega_P,\omega_S)^2=|\det(\I-P-S)|^{1/2}
\end{equation}
is valid.
\end{thm}
%%%%%%%%%%%%%%%%%%%%%%%%%%%%%%%%%%%%%%%%%%%%%%%%%%%%%%%%%%%%%%%%%%%%%%%%%%%%%%%%%%%%%%%%%%%%%%%%%%%%
\begin{proof}
Since $\omega_P$ is a pure state, the square fidelity $F(\omega_P,\omega_S)^2=\omega_S(E_P)$ is
given by the expectation value of the support projection $E_P$ in the state $\omega_S$.
To calculate the fidelity, we take advantage of Theorem~\ref{thm:convol-four} and
Theorem~\ref{thm:g-extension} which can be used to express
the expectation value of the support projection $E_P$ in the state $\omega_S$ as a Gaussian
Grassmann integral 
\begin{equation}
\omega_S(E_P)\I
=\epsilon_{[v,P]}\int_Q v(\xi) \
\8e^{\frac{1}{2}\RIGG{\xi^\ADJ}{(P-S)\xi}{Q}}\8e^{-\RIGG{\xi^\ADJ}{P\xi}{Q}}
=\epsilon_{[v,P]}\int_Qv(\xi) \ \8e^{\frac{1}{2}\RIGG{\xi^\ADJ}{(\I-P-S)\xi}{Q}}
\end{equation}
where $v$ is a normalized selfadjoint form of highest degree. We evaluate the integral with help of
(\ref{equ:gaussian}) which leads to
\begin{equation}
\omega_S(E_P)
=\epsilon_{[v,P]}\8{Pf}_{[v,J]}(\I-P-S) \; .
\end{equation}
Since the left hand side is positive (expectation value of a positive operator) and by the identity 
$\8{Pf}_{[v,J]}(\I-P-S)^2=\det(\I-P-S)$ we obtain the desired result.
\end{proof}
%%%%%%%%%%%%%%%%%%%%%%%%%%%%%%%%%%%%%%%%%%%%%%%%%%%%%%%%%%%%%%%%%%%%%%%%%%%%%%%%%%%%%%%%%%%%%%%%%%%%

%%%%%%%%%%%%%%%%%%%%%%%%%%%%%%%%%%%%%%%%%%%%%%%%%%%%%%%%%%%%%%%%%%%%%%%%%%%%%%%%%%%%%%%%%%%%%%%%%%%%
\subsection{Finite versus infinite dimensions}
%%%%%%%%%%%%%%%%%%%%%%%%%%%%%%%%%%%%%%%%%%%%%%%%%%%%%%%%%%%%%%%%%%%%%%%%%%%%%%%%%%%%%%%%%%%%%%%%%%%%
\label{subsec:finite-infinite}
%%%%%%%%%%%%%%%%%%%%%%%%%%%%%%%%%%%%%%%%%%%%%%%%%%%%%%%%%%%%%%%%%%%%%%%%%%%%%%%%%%%%%%%%%%%%%%%%%%%%
To what extend can our formalism be used for infinite
dimensional fermion systems? Recall that the dimension of the fermion system is given here by the
dimesion of the projection $Q$. The GAR algebra $\GAR{H}{J}{Q}$ is well defined for all separable
Hilbert spaces and for all projections $Q$ commuting with $J$. The anticommutative phase space for
fermions as well as the Weyl-Grassmann operators can also be constructed in this case. The main
difficulty to extend our kind of ``harmonic analysis'' to infinite dimensional systems is due to the
problem of defining an appropriate infinite dimensional Grassmann integral. Of course, one can try
to approximate (in some appropriate sense) an infinite Grassmann integral by a sequence of finite
dimensional integrations. For instance, one can choose an increasing sequence of projections
$(E_n)$, $\lim E_n=\I$,  that commute with $Q$ and $J$ and suppose further that $\8{dim}(E_nQ)=2n$.
Now each of the the anticommutative phase spaces $\PHSP{E_nH}{E_nJ}{E_nQ}$ can be identified with a
subspace of $\PHSP{H}{J}{Q}$. For infinite rank projectons $Q$ we propose to call a function $f$ on
$\PHSP{H}{J}{Q}$ with values in a right module $\9E$ to be G-holomorphic if for any finite rank
projection $E$ that commutes with $Q$ and $J$ the function 
$\EQ{f}{EQ}\mathpunct:\PHSP{EH}{EJ}{EQ}\ni\xi\to f(\xi)$ is G-holomorphic. A proposal for an
infinite Grassmann integral is now to choose an appropriate sequence of forms
$v=(v_n)$, where $v_n$ is of highest degree with respect to $E_nQ$, and define
\begin{equation}
\label{eq:inf-dim-int}
\int_Q v(\xi) \ f(\xi) :=\lim_{n} \int_{E_nQ} v_n(\xi) \ f(\xi)
\end{equation}
for an appropriate limit. The problem is here to find a reasonable form of convergence.

Concerning applications, one can directly ask whether a statement like Theorem~\ref{thm:fidelity}
can expected to be true for infinite dimensional fermion systems. Let $P$ be a basis projection
and $S$ be a covariance operator on an infinite dimensional Hilbert space $K=QH$. Then the fidelity
is only non-vanishing if the GNS representations of $\omega_S$ and $\omega_P$ are mutually quasi
equivalent. To explain this, let $\pi$ denote the GNS representation of $\omega_P$, let $\2H$ denote
the GNS Hilbert space and let $\Omega$ be the GNS
(Fock vacuum) vector in $\2H$, i.e. $\omega_P(A)=\SPROD{\Omega}{\pi(A)\Omega}$ holds for all
operators $A$ in the fermion algebra. The quasi equivalence implies that there exists a
density operator $\rho$ on $\2H$ such that $\omega_S(A)=\8{tr}(\rho\pi(A))$ holds.
A necessary and sufficient criterion for quasi equivalence has been shown by Araki
\cite[Theorem 1]{Araki87} which states that the operator $\sqrt{S}-P$ has to be a
Hilbert-Schmidt operator. A first natural question that arises here is whether the determinant
$\det(\I-P-S)$ exists under these circumstances. Indeed the determinant of an operator $\I+A$
exists in infinite dimensions, provided $A$ is a trace class operator (see for instance
\cite{RESI4} Chapter XIII). Unfortunately, the operator $P+S$ is in general {\em not} trace class
and the formula for the fidelity in Theorem~\ref{thm:fidelity} cannot hold in infinite
dimensions. One possible way out of this dilemma is to ``renormalize'' the determinant with respect
to the basis projection $P$. For instance, we multiply the operator $\I-P-S$ by the unitary $\I-2P$
from the right which yields $\I-P-S+2SP$. In finite dimensions the modulus of the determinant
$|\det(\I- P-S)|=|\det(\I-P-S+2SP)|$ does not change and the determinant $\det(\I-P-S+2SP)$ is
normalized in the sense that for $S=P$ we get 
$\det(\I-P-S+2SP)=\det(\I)=1$. Now, if $\sqrt{S}-P$ is Hilbert-Schmidt, then it follows that
$P+S-\ACO{P}{S}$ is trace class. Thus for the simplified case that $[P,S]=0$ the ``renormalized''
determinant $\det(\I-P-S+2SP)$ exists in infinite dimensions. This shows that there might be an
analogous determinant formula for the fidelity between a pure quasi-free and an arbitrary
quasi-free state for the infinite dimensional case.  
%%%%%%%%%%%%%%%%%%%%%%%%%%%%%%%%%%%%%%%%%%%%%%%%%%%%%%%%%%%%%%%%%%%%%%%%%%%%%%%%%%%%%%%%%%%%%%%%%%%%

%%%%%%%%%%%%%%%%%%%%%%%%%%%%%%%%%%%%%%%%%%%%%%%%%%%%%%%%%%%%%%%%%%%%%%%%%%%%%%%%%%%%%%%%%%%%%%%%%%%%
\section*{Acknowledgments}
The project COQUIT acknowledges the financial support of the Future and Emerging Technologies
(FET) programme within the Seventh Framework Programme for Research of the European Commission,
under FET-Open grant number: 233747.

%%%%%%%%%%%%%%%%%%%%%%%%%%%%%%%%%%%%%%%%%%%%%%%%%%%%%%%%%%%%%%%%%%%%%%%%%%%%%%%%%%%%%%%%%%%%%%%%%%%%
\section{On the structure of the GAR algebra}
%%%%%%%%%%%%%%%%%%%%%%%%%%%%%%%%%%%%%%%%%%%%%%%%%%%%%%%%%%%%%%%%%%%%%%%%%%%%%%%%%%%%%%%%%%%%%%%%%%%%
\label{sec:math_struc}
%%%%%%%%%%%%%%%%%%%%%%%%%%%%%%%%%%%%%%%%%%%%%%%%%%%%%%%%%%%%%%%%%%%%%%%%%%%%%%%%%%%%%%%%%%%%%%%%%%%%
In this section, we discuss mathematical issues on the GAR algebra which are needed to prepare
and derive the results which we have discussed in in Subsection~\ref{subsec:main-results}. Hereby
we mainly focus on the algebraic and functional analytic properties. All the results, that we
derive here are also valid for the infinite dimensional (but separable) case. 

%%%%%%%%%%%%%%%%%%%%%%%%%%%%%%%%%%%%%%%%%%%%%%%%%%%%%%%%%%%%%%%%%%%%%%%%%%%%%%%%%%%%%%%%%%%%%%%%%%%%
\subsection{Existence of the adjoint}
%%%%%%%%%%%%%%%%%%%%%%%%%%%%%%%%%%%%%%%%%%%%%%%%%%%%%%%%%%%%%%%%%%%%%%%%%%%%%%%%%%%%%%%%%%%%%%%%%%%%
%%%%%%%%%%%%%%%%%%%%%%%%%%%%%%%%%%%%%%%%%%%%%%%%%%%%%%%%%%%%%%%%%%%%%%%%%%%%%%%%%%%%%%%%%%%%%%%%%%%%
\label{subsec:adjoint}
%%%%%%%%%%%%%%%%%%%%%%%%%%%%%%%%%%%%%%%%%%%%%%%%%%%%%%%%%%%%%%%%%%%%%%%%%%%%%%%%%%%%%%%%%%%%%%%%%%%%
As we have introduced in Subsection~\ref{subsec:GAR}, the GAR algebra for a Hilbert space $H$, a
complex conjugation $J$ and a projection $Q$ that commutes with $J$ is defined as the norm closed
subalgebra of the enveloping CAR algebra $\ENV{H}{J}{Q}$ that is generated by the operators
$G(f):=B(f\oplus 0)$. Therefore, by construction, the GAR algebra is a Banach algebra. As we have
promised in Subsection~\ref{subsec:GAR}, we show here that the  GAR algebra admits a continuous
adjoint that coincides with the C*-adjoint on the fermionic part. 
 
%%%%%%%%%%%%%%%%%%%%%%%%%%%%%%%%%%%%%%%%%%%%%%%%%%%%%%%%%%%%%%%%%%%%%%%%%%%%%%%%%%%%%%%%%%%%%%%%%%%%
\begin{prop}
\label{prop:adj}
\begin{enumerate}
\item 
There exists an adjoint $^\ADJ\mathpunct:\GAR{H}{J}{Q}\to\GAR{H}{J}{Q}$ such that the GAR algebra
becomes a Banach *-algebra. 
\item
The adjoint $^\ADJ$ coincides with the C*-adjoint on the fermionic part: $A^\ADJ=A^*$
for $A\in\CAR{QH}{QJ}$. 
\item
The adjoint $^\ADJ$ is uniquely determined by the relation $G(f)^\ADJ=G(Jf)$ for $f\in H$. 
\end{enumerate}
\end{prop}
%%%%%%%%%%%%%%%%%%%%%%%%%%%%%%%%%%%%%%%%%%%%%%%%%%%%%%%%%%%%%%%%%%%%%%%%%%%%%%%%%%%%%%%%%%%%%%%%%%%%
\begin{proof}
We consider the dense subalgebra $\PGAR{H}{J}{Q}$ in $\GAR{H}{J}{Q}$ that consists of finite sums of
finite products of operators $G(f)$ with $f\in H$. On this dense subalgebra, we define an
anti-linear involution according to $G(f)^\ADJ=G(Jf)$ for $f\in H$. To prove the proposition, we
just have to show that this involution is continuous. For this purpose, we consider a finite rank
projection $E$ that commutes with $Q$ and $J$. We obtain a closed finite dimensional subalgebra
$\GAR{EH}{EJ}{EQ}\subset\PGAR{H}{J}{Q}$ which is closed under the adjoint $^\ADJ$. Therefore,
$^\ADJ$ is bounded on $\GAR{EH}{EJ}{EQ}$ which implies $\|A^\ADJ\|=\|A\|$ for all $A\in
\GAR{EH}{EJ}{EQ}$. Note that any bounded involution on a Banach space is isometric. Since each
operator $A$ in the GAR algebra is a norm convergent limit of a sequence of operators $(A_n)$, where
$A_n\in \GAR{E_nH}{E_nJ}{E_nQ}$ and $(E_n)$ is an increasing sequence of projections that commute
with $J$ and $Q$, we conclude that the adjoint $^\ADJ$ is isometric on a norm dense subalgebra. 
Hence it is norm continuous and can uniquely be extended to the full GAR algebra.
Since we have $G(Qf)^*=B(J_Q(Qf\oplus 0))=G(JQf)=G(Qf)^\ADJ$, we conclude that $^\ADJ$ coincides
with the C*-adjoint on the fermionic part. 
\end{proof}
%%%%%%%%%%%%%%%%%%%%%%%%%%%%%%%%%%%%%%%%%%%%%%%%%%%%%%%%%%%%%%%%%%%%%%%%%%%%%%%%%%%%%%%%%%%%%%%%%%%%

%%%%%%%%%%%%%%%%%%%%%%%%%%%%%%%%%%%%%%%%%%%%%%%%%%%%%%%%%%%%%%%%%%%%%%%%%%%%%%%%%%%%%%%%%%%%%%%%%%%%
\subsection{Natural norms on the Grassmann algebra}
%%%%%%%%%%%%%%%%%%%%%%%%%%%%%%%%%%%%%%%%%%%%%%%%%%%%%%%%%%%%%%%%%%%%%%%%%%%%%%%%%%%%%%%%%%%%%%%%%%%%
\label{subsec:norms-on-Grassmann}
%%%%%%%%%%%%%%%%%%%%%%%%%%%%%%%%%%%%%%%%%%%%%%%%%%%%%%%%%%%%%%%%%%%%%%%%%%%%%%%%%%%%%%%%%%%%%%%%%%%%
As a Banach *-algebra, the natural norm on the Grassmann algebra is the operator norm that is
induced by the enveloping CAR algebra. In this norm the product is continuous which is the defining
property of a Banach algebra norm. 

As we have discussed in Subsection~\ref{subsec:GAR}, the
anti-symmetric tensor algebra $\ATEN{H}{J}$ over the Hilbert space $H$ with complex conjugation
$J$ can be identified with a norm dense *-subalgebra in $\GR{H}{J}$. In particular, $\ATEN{H}{J}$
can also be identified with a dense subspace of the anti-symmetric Fock space $\AFOCK{H}$. Thus
$\ATEN{H}{J}$ possesses a scalar product $(\lambda,\mu)\mapsto\SPROD{\lambda}{\mu}$. The
norm $\HNO{\lambda}=\sqrt{\SPROD{\lambda}{\lambda}}$, which is induced by the scalar product, will
be called the {\em Fock space norm}. The first observation is, that the Fock space norm, and
hence the scalar product, is continuous with respect to the Banach algebra norm as stated by the
following proposition: 

%%%%%%%%%%%%%%%%%%%%%%%%%%%%%%%%%%%%%%%%%%%%%%%%%%%%%%%%%%%%%%%%%%%%%%%%%%%%%%%%%%%%%%%%%%%%%%%%%%%%
\begin{prop}
\label{prop:norm-cont}
The norm $\HNO{\cdot}$ is continuous with respect to the underlying
Banach algebra norm $\|\cdot\|$. In particular, the scalar product is continuous and can uniquely be
extended to the full Grassmann algebra.
\end{prop}
%%%%%%%%%%%%%%%%%%%%%%%%%%%%%%%%%%%%%%%%%%%%%%%%%%%%%%%%%%%%%%%%%%%%%%%%%%%%%%%%%%%%%%%%%%%%%%%%%%%%
\begin{proof}
By the identity $\ATEN{H}{J}=\PGAR{H}{J}{0}$ the enveloping CAR algebra is given by
$\CAR{H\oplus H}{J_0}$, where $J_0$ is the complex conjugation $J_0(h\oplus f)=Jf\oplus Jh$. We
introduce athe basis projection $E$ on $H\oplus H$ according to $E(f\oplus
h):=f\oplus 0$. The corresponding Fock representation of the enveloping
C*-algebra $\pi$ on $\AFOCK{H}$ is faithful. The Grassmann field
operators are represented in terms of the creation operators by
$\pi(\Lambda(f))=\pi(B(f\oplus 0))=c^*(f)$. Clearly, the vectors $c^*(f_1)\cdots
c^*(f_n)\Omega$ span a dense subspace in $\AFOCK{H}$, where $\Omega$ is the corresponding Fock
vacuum vector. Therefore we have for an operator
$\lambda\in\ATEN{H}{J}$:
\begin{equation}
\HNO{\lambda}=\|\pi(\lambda)\Omega\|\leq \|\lambda\| \; .
\end{equation}
Hence the norm $\HNO{\cdot}$ is continuous with respect to $\|\cdot\|$.
\end{proof}
%%%%%%%%%%%%%%%%%%%%%%%%%%%%%%%%%%%%%%%%%%%%%%%%%%%%%%%%%%%%%%%%%%%%%%%%%%%%%%%%%%%%%%%%%%%%%%%%%%%%

It is worth to mention that both norms are different from each other. In particular, the norm
$\HNO{\cdot}$ is {\em not} a Banach algebra norm. This can be verified by the following counter
example: Take mutually orthogonal vectors $e_1,\cdots, e_6$ in $H$ and consider the operator
$\lambda=e_1\wedge e_2+e_3\wedge e_4+e_5\wedge e_6$. Then we find for the Fock space norm
$\HNO{\lambda}=\sqrt{3}$ whereas a straight forward computation for the Fock space norm of
$\lambda^2$ yields $\HNO{\lambda^2}=2\sqrt{3}>3=\HNO{\lambda}^2$.

%%%%%%%%%%%%%%%%%%%%%%%%%%%%%%%%%%%%%%%%%%%%%%%%%%%%%%%%%%%%%%%%%%%%%%%%%%%%%%%%%%%%%%%%%%%%%%%%%%%%
\subsection{Nilpotent ideals, positive operators and positive functionals}
\label{subsec:ideals,positivity}
%%%%%%%%%%%%%%%%%%%%%%%%%%%%%%%%%%%%%%%%%%%%%%%%%%%%%%%%%%%%%%%%%%%%%%%%%%%%%%%%%%%%%%%%%%%%%%%%%%%%
As a *-algebra, $\GAR{H}{J}{Q}$ possesses a natural convex cone of positive elements: The set of
positive linear functionals consits of all linear functionals
$\omega\mathpunct:\GAR{H}{J}{Q}\to\7C$ with $\omega(A^\ADJ A)\geq 0$. The positive cone
$\GAR{H}{J}{Q}_+$ consits of all operators that have positive expectation values for all positive
functionals.

In order to analyze the positivity of operators, we introduce the norm closed two-sided ideal
$\IDL{H}{J}{Q}$ that is genrated by the {\em selfadjoint nilpotent operators} that belong to the
Grassmann part $\GR{Q^\perp H}{Q^\perp J}$. Recall that an operator $Z$ is nilpotent if there exists
$n\in\7N$ with $Z^n=0$. The operators that are given by finite sums 
\begin{equation}
A=\sum_{i=0}^k A_i Z_i
\end{equation}
with selfadjoint nilpotent operators $Z_i\in\GR{Q^\perp H}{Q^\perp J}$ and $A_i\in \GAR{H}{J}{Q}$
form a norm dense subset in $\IDL{H}{J}{Q}$. It s not difficult to see that $A,A^\ADJ$ as well as
$A+A^\ADJ$ are nilpotent operators in the GAR algebra. Here one takes advantage of the fact, that
the graded commutator of any operator with an operator from the Grassmann part vanishes. 

%%%%%%%%%%%%%%%%%%%%%%%%%%%%%%%%%%%%%%%%%%%%%%%%%%%%%%%%%%%%%%%%%%%%%%%%%%%%%%%%%%%%%%%%%%%%%%%%%%%%
\begin{prop}
\label{prop:positivity}
\begin{enumerate}
\item 
There exists a surjective *-homomorphism $\epsilon_Q\mathpunct:\GAR{H}{J}{Q}\to\CAR{QH}{QJ}$ such
that the identities $\epsilon_Q(Z)=0$ and $\epsilon_Q(A)=A$ are valid for all $Z\in\IDL{H}{J}{Q}$
and for all $A\in\CAR{QH}{QJ}$.

\item
To each positive functionals $\omega$ on the GAR algebra $\GAR{H}{J}{Q}$ there exists a unique
positive functional $\omega'$ on the fermionic part $\CAR{QH}{QJ}$ such that 
\begin{equation}
\omega=\omega'\circ\epsilon_Q \; .
\end{equation}
\end{enumerate}
\end{prop}
%%%%%%%%%%%%%%%%%%%%%%%%%%%%%%%%%%%%%%%%%%%%%%%%%%%%%%%%%%%%%%%%%%%%%%%%%%%%%%%%%%%%%%%%%%%%%%%%%%%
\begin{proof}
\begin{enumerate}
\item 
In order to prove the existence of $\epsilon_Q$, we show that the quotient algebra
$\GAR{H}{J}{Q}/\IDL{H}{J}{Q}$ is canonically isomorphic to $\CAR{QH}{QJ}$. Let $\pi_Q$ be the
canonical *-homomorphism that projects $\GAR{H}{J}{Q}$ onto $\GAR{H}{J}{Q}/\IDL{H}{J}{Q}$. The
intersection $\CAR{QH}{QJ}\cap\IDL{H}{J}{Q}=\{0\}$ only contains the zero element. Namely, let $Z$
be selfadjoint and nilpotent, then $Z\in \CAR{QH}{QJ}$ implies $Z=0$ since the
only selfadjoint and nilpotent element inside a C*-algebra is the zero operator. For each generator
$G(f)$ of the GAR algebra we have the decomposition $G(f)=G(Qf)+G(Q^\perp f)$ with $G(Q^\perp
f)\in\IDL{H}{J}{Q}$. Thus we conclude $\pi_Q(G(f))=G(Qf)+\IDL{H}{J}{Q}$ and
$\iota_Q\mathpunct:G(Qf)+\IDL{H}{J}{Q}\mapsto
B(Qf)\in\CAR{QH}{QJ}$ is the desired isomorphism. Thus $\epsilon_Q:=\iota_Q\circ\pi_Q$ is a
*-homomorphism that annihilates the ideal $\IDL{H}{J}{Q}$ and acts as the identity on the fermionic
part.

\item
Let $\omega$ be a positive functional on the GAR algebra and let $Z$ be selfadjoint and nilpotent.
Then we can choose $k\in \7N$ such that $Z^n=0$ with $n=2^k$. By iterating the Cauchy-Schwarz
inequality, we conclude $|\omega(Z)|\leq\omega(Z^n)=0$. Thus $\omega$ annihilates all nilpotent
selfadjoint elements. Since the ideal $\IDL{H}{J}{Q}$ possesses a dense subsubspace that is spanned
by nilpotent selfadjoint operators and since $\omega$ is continuous, the ideal $\IDL{H}{J}{Q}$ is
annihilated which implies $\omega=\omega'\circ\epsilon_Q$, where $\omega'$ is the restriction of
$\omega$ to the fermionic part. Thus each positive functional on the GAR algebra is the pull back
of a unique (note that the dual map of $\epsilon_Q$ is injective) positive functional on the
fermionic part via the *-homomorphism $\epsilon_Q$. By construction we have
$\omega(A+Z)=\omega'(\epsilon_Q(A+Z))=\omega'(A)$. 
\end{enumerate}
\end{proof}
%%%%%%%%%%%%%%%%%%%%%%%%%%%%%%%%%%%%%%%%%%%%%%%%%%%%%%%%%%%%%%%%%%%%%%%%%%%%%%%%%%%%%%%%%%%%%%%%%%%%

%%%%%%%%%%%%%%%%%%%%%%%%%%%%%%%%%%%%%%%%%%%%%%%%%%%%%%%%%%%%%%%%%%%%%%%%%%%%%%%%%%%%%%%%%%%%%%%%%%%%
\section{On Grassmann integrals}
\label{sec:grasscalc}
%%%%%%%%%%%%%%%%%%%%%%%%%%%%%%%%%%%%%%%%%%%%%%%%%%%%%%%%%%%%%%%%%%%%%%%%%%%%%%%%%%%%%%%%%%%%%%%%%%%%
Towards the development of a ``harmonic analysis'' on antisymmetric phase space, we review here
the basic concepts of Grassmann calculus, including integration theory. In view of our applications
to fermionic systems, we need to give here a version which at some points differ from the standard
analysis that can be found within the literature. In what follows, we assume here that the
projection $Q$ is of even finite rank $=2n$, and that the rank of
$Q^\perp$ is infinite.

%%%%%%%%%%%%%%%%%%%%%%%%%%%%%%%%%%%%%%%%%%%%%%%%%%%%%%%%%%%%%%%%%%%%%%%%%%%%%%%%%%%%%%%%%%%%%%%%%%%%
\subsection{G-holomorphic functions}
%%%%%%%%%%%%%%%%%%%%%%%%%%%%%%%%%%%%%%%%%%%%%%%%%%%%%%%%%%%%%%%%%%%%%%%%%%%%%%%%%%%%%%%%%%%%%%%%%%%%
\label{subsec:G-hol}
%%%%%%%%%%%%%%%%%%%%%%%%%%%%%%%%%%%%%%%%%%%%%%%%%%%%%%%%%%%%%%%%%%%%%%%%%%%%%%%%%%%%%%%%%%%%%%%%%%%%
In order to perform integration with respect to Grassmann variables we have to consider the 
algebra of functions that can be integrated. These functions are appropriate polynomials of
Grassmann variables $\xi\in\PHSP{H}{J}{Q}_0$ with values in a right module over the ring
$\GR{Q^\perp H}{Q^\perp J}$. If the underlying ring structure is clear from the context we just
briefly say ``right module''. In the following, the right modules $\9E$ under consideration are
assumed to be Banach spaces with a continuous right multiplication.

Roughly, G-holomorphic functions are polynomials in the Grassmann variables. The problem is, that
the polynomial representation is not unique which causes ambiguities in the definition of the
Grassmann integral. In order to overcome this problem, we define $\9E_Q$ to be the norm closed
sub-rightmodule that consits of all elements $X$ such that $X\lambda_1\cdots\lambda_{2n}=0$
holds for all families $\lambda_1,\cdots, \lambda_{2n}$ of $2n$ odd operators in the Grassmann
part ($2n$ is the rank of the projection $Q$). As already mentioned, the definition of the
Grassmann integral is most comfortable in the case where the submodule $\9E_Q=\{0\}$ is trivial. 
In this context, the most important example for such a right module is the GAR algebra for
which the complementry projection $Q^\perp$ has infinite rank. The Grassmann part
is isomorphic to the DeWitt algebra that is build from an infinite number of
anticommuting generators \cite{ChoBruDeWitt89}.

%%%%%%%%%%%%%%%%%%%%%%%%%%%%%%%%%%%%%%%%%%%%%%%%%%%%%%%%%%%%%%%%%%%%%%%%%%%%%%%%%%%%%%%%%%%%%%%%%%%%
\begin{prop}
\label{prop:reduce}
Let $Q$ be a projection of even and finite rank $2n$ and suppose that the rank of $Q^\perp$ is
infinite. Then the closed subspace $\GAR{H}{J}{Q}_Q=\{0\}$ is trivial. 
\end{prop}
%%%%%%%%%%%%%%%%%%%%%%%%%%%%%%%%%%%%%%%%%%%%%%%%%%%%%%%%%%%%%%%%%%%%%%%%%%%%%%%%%%%%%%%%%%%%%%%%%%%%
\begin{proof}
Let $E$ be  a finite rank projection that commutes with $Q$ and $J$. Then the GAR algebra
$\GAR{EH}{EJ}{EQ}$ is a finite dimensional subalgebra of $\GAR{H}{J}{Q}$. Using the enveloping CAR
algebra, the full GAR algebra is isomrphic to the twisted tensor product
$\GAR{H}{J}{Q}=\GAR{EH}{EJ}{EQ}\tilde\otimes\GAR{E^\perp H}{E^\perp J}{E^\perp J}$ where
$\tilde\otimes$ denotes the twisted tensor product. Note that the Banach algebra norm is a cross
norm with respect to the twisted tensor product. Since $Q^\perp$ is infinite dimensional and $E$
is finite dimensional, we conclude that $E^\perp Q^\perp=Q^\perp E^\perp$ is infinite dimensional.
For each $A\in\GAR{EH}{EJ}{EQ}$ with $\|A\|>0$ we can choose odd Grassmann operators
$\lambda_1,\cdots,\lambda_{2n}\in\GR{E^\perp Q^\perp H}{E^\perp Q^\perp J}_1$ such that
$\|\lambda_1\cdots\lambda_{2n}\|=1$. This implies that
$\|A\lambda_1\cdots\lambda_{2n}\|=\|A\|\|\lambda_1\cdots\lambda_{2n}\|=\|A\|$. Here we have used the
fact that $A\lambda_1\cdots\lambda_{2n}\cong A\tilde\otimes\lambda_1\cdots\lambda_{2n}$. For each
non-zero $A\in\GAR{H}{J}{Q}$ and for each $\epsilon$ with $\|A\|>\epsilon>0$ we can find a finite
rank projection $E$ that commutes with $Q$ and $J$ and an operator $A_\epsilon\in \GAR{EH}{EJ}{EQ}$
with $\|A_\epsilon\|=\|A\|$ such that $\|A-A_\epsilon\|\leq \epsilon$. Again we can find
odd Grassmann operators $\lambda_1,\cdots,\lambda_{2n}\in\GR{E^\perp Q^\perp H}{E^\perp Q^\perp
J}_1$ such that
$\|\lambda_1\cdots\lambda_{2n}\|=1$. Suppose now that $A\lambda_1\cdots\lambda_{2n}=0$ then we
conclude $\|A\lambda_1\cdots\lambda_{2n}-A_\epsilon\lambda_1\cdots\lambda_{2n}
\|=\|A_\epsilon\lambda_1\cdots\lambda_{2n}\|=\|A\|\leq \epsilon$ which contradicts the assumption
$\|A\|>\epsilon$. Therefore $A\lambda_1\cdots\lambda_{2n}\not=0$ which implies that the subspace
$\GAR{H}{J}{Q}_Q=\{0\}$ is trivial.
\end{proof}
%%%%%%%%%%%%%%%%%%%%%%%%%%%%%%%%%%%%%%%%%%%%%%%%%%%%%%%%%%%%%%%%%%%%%%%%%%%%%%%%%%%%%%%%%%%%%%%%%%%%

The vector space $\GHOL{H}{J}{Q}{\9E}$ of {\em G-holomorphic functions} with values in a right
module $\9E$ consists of all functions from $\PHSP{H}{J}{Q}_0$ into $\9E$ which can be
build from linear combinations of monomial functions
\begin{equation}
\xi\mapsto X\xi_{u_1}\cdots\xi_{u_n}
\end{equation}
with $X\in\9E$ and $u_1,\cdots, u_n\in QH$. Here the ``$u$-component'' of $\xi$ is  defined as
$\xi_u:=\RIGG{u\otimes\I}{\xi}{Q}$. The algebra (ring) of ``Grassmann-valued'' G-holomorphic
functions $\RHOL{H}{J}{Q}$ is
the algebra that is generated by the functions $\xi\mapsto \xi_u$ with $u\in QH$. The vector space,
as defined above, is canonically equipped with a right module structure. For a G-holomorphic
function we define the corresponding action by $(F\cdot\lambda)(\xi):=F(\xi)\lambda$. The
G-holomorphic functions $\GHOL{H}{J}{Q}{\9E}$ with values in the right module $\9E$ are also
equipped with a $\RHOL{H}{J}{Q}$ right module structure. Indeed, the space $\GHOL{H}{J}{Q}{\9E}$ is
the right module over $\RHOL{H}{J}{Q}$, generated by $\9E$. Each G-holomorphic function admits a
polynomial representation induced by a real orthonormal basis $(e^i)_{i\in N}$ of $QH$ that
is indexed by the ordered set
$N=\{1,\cdots,2n\}$. Any vector $\xi\in\PHSP{H}{J}{Q}_0$ can be expanded
in this basis as $\xi=\sum_i\Lambda^i\xi_i$ with $\xi_i\in\GR{Q^\perp H}{Q^\perp J}_1$. The
corresponding polynomial expansion of a G-holomorphic function $F$ is given by 
\begin{equation}
F(\xi)=\sum F^I \ \xi_I
\end{equation}
with coefficients $F^I$ in the right module $\9E$. The monomial $\xi_I$ which is
associated to an ordered subset $I=\{i_1<i_2<\cdots<i_k\}\subset N$ is given by
$\xi_I:=\xi_{i_1}\cdots\xi_{i_k}$. In particular, since $Q$ has finite rank $2n$, each G-homomorphic
function can be expressed as a finite sum of monomials, i.e., there is no problem
concerning convergence. 

There is an interesting connection between G-holomorphic
functions and right module homomorphisms. To make this point clear, we observe that the Grassmann
algebra
$\GR{H}{J}$ possesses a natural right module structure over $\GR{Q^\perp H}{Q^\perp J}$ by
right multiplication $a\mapsto a\lambda$ with $a\in\GR{H}{J}$ and $\lambda\in\GR{Q^\perp H}{Q^\perp
J}$. We denote by $\RHOM{H}{J}{Q}{\9E}$ the Banach space of bounded right module homomorphisms $R$
from $\GR{H}{J}$ into $\9E$, i.e. $R$ is complex linear and fulfills the condition
$R(a\lambda)=R(a)\lambda$ for $a\in\GR{H}{J}$ and $\lambda\in\GR{Q^\perp H}{Q^\perp J}$. A
particular case is here the Banach space of the right module homomorphisms with values in
$\GR{Q^\perp H}{Q^\perp J}$ which will be denoted by $\REND{H}{J}{Q}$. Clearly, the space of right
module homomorphism is a right module itself according to the following definition: An
operator $\lambda\in\GR{Q^\perp H}{Q^\perp J}$ acts on a right module homomorphism $R$
as $(R\cdot\lambda)(a):=R(\lambda a)$. 

For a G-holomorphic function $F\in\GHOL{H}{J}{Q}{\9E}$ we consider the subset
$\9R_QF\subset\RHOM{H}{J}{Q}{\9E}$ that consists of all right module homomorphisms $R$ such that
$F(\xi)=R(\8e^\xi)$ holds.  To prepare our the definition of the Grassmann integral, we consider an
operator $v\in\GR{QH}{QJ}$ is called {\em a form of highest degree
with respect to $Q$} if $v \Lambda(h)=0$ for all $h\in QH$.  If it is clear from the
context to which projection $Q$ we are referring, we shortly say that $v$ is a form of highest
degree. In general we say that an operator is {\em a $k$-form with respect to $Q$} if it is a
linear combination of operators of the form
$\Lambda(Qf_1)\cdots\Lambda(Qf_k)=Qf_1\wedge\cdots\wedge Qf_k$. Let $n$ be the rank of the
projection $Q$, then the subspace of $k$-forms has dimension $n \choose k$ and the 
subspace of forms of highest degree ($n$-forms) is a one-dimensional.

%%%%%%%%%%%%%%%%%%%%%%%%%%%%%%%%%%%%%%%%%%%%%%%%%%%%%%%%%%%%%%%%%%%%%%%%%%%%%%%%%%%%%%%%%%%%%%%%%%%%
\begin{prop}
\label{prop:unique}
Let $F\in\GHOL{H}{J}{Q}{\9E}$ be a G-holomorphic function and let 
$R_1,R_2\in\9R_QF\subset\RHOM{H}{J}{Q}{\9E}$ be two right module homomorphisms. Moreover, we
assume that the submodule $\9E_Q=\{0\}$ is trivial. Then the identity 
\begin{equation}
\label{equ:quotient}
R_1(v)=R_2(v)
\end{equation}
holds for all forms $v$ of highest degree. 
\end{prop}
%%%%%%%%%%%%%%%%%%%%%%%%%%%%%%%%%%%%%%%%%%%%%%%%%%%%%%%%%%%%%%%%%%%%%%%%%%%%%%%%%%%%%%%%%%%%%%%%%%%%
\begin{proof}
Let $(e^i)_{i\in N}$ be a real basis of $QH$, indexed be the ordered set $N=\{1,2,\cdots,n\}$. For
each ordered subset $I\subset N$ we introduce the operator
$\Lambda^I=\Lambda(e^{i_k})\cdots\Lambda(e^{i_1})$ with $I=\{i_1<i_2<\cdots <i_k\}$ where we put
$\Lambda^\emptyset=\I$. Since $(\Lambda^I)_{I\subset N}$ is a basis of $\GR{QH}{QJ}$, the operator
$\8e^\xi$ can be expanded as $\sum \Lambda^I\xi_I$ with $\xi_I=\xi_{i_1}\cdots\xi_{i_k}$. It
follows from $F(\xi)=R_1(\8e^\xi)=R_2(\8e^\xi)$, that we obtain for the right module homomorphism
$D=R_1-R_2$
\begin{equation}
\sum_{I\subset N}D(\Lambda^I)\xi_I=0
\end{equation}
for all $\xi\in \PHSP{H}{J}{Q}$. From this we conclude that for each $I\subset N$ the identity 
\begin{equation}
D(\Lambda^I)\xi_I=0
\end{equation}
holds for all $\xi\in \PHSP{H}{J}{Q}$. Since $v$ is an operator of highest degree with respect to
$Q$ we have $v=v_N\Lambda^N$ which implies (\ref{equ:quotient}).
\begin{equation}
D(v)\xi_N=0
\end{equation}
for all $\xi$. But then $D(v)$ is contained in the submodule in $\9E_Q=\{0\}$ which implies
(\ref{equ:quotient}).
\end{proof}
%%%%%%%%%%%%%%%%%%%%%%%%%%%%%%%%%%%%%%%%%%%%%%%%%%%%%%%%%%%%%%%%%%%%%%%%%%%%%%%%%%%%%%%%%%%%%%%%%%%%

%%%%%%%%%%%%%%%%%%%%%%%%%%%%%%%%%%%%%%%%%%%%%%%%%%%%%%%%%%%%%%%%%%%%%%%%%%%%%%%%%%%%%%%%%%%%%%%%%%%%
\subsection{Definition of the Grassmann integral and some basic properties}
%%%%%%%%%%%%%%%%%%%%%%%%%%%%%%%%%%%%%%%%%%%%%%%%%%%%%%%%%%%%%%%%%%%%%%%%%%%%%%%%%%%%%%%%%%%%%%%%%%%%
In the following discussion, we only consider right modules $\9E$ for which the the submodule
$\9E_Q=\{0\}$ is trivial. Let $v\not=0$ be an non-zero form of highest degree. The Grassmann
integral of $F$ with respect to
$v$ is defined according to  
\begin{equation}
\int_Q v(\xi) \ F(\xi):= R(v)
\end{equation}
with an right module homomorphism $R\in \9R_QF$. Note that by Proposition~\ref{prop:unique}, this
definition only depends on the G-holomorphic function itself. The notation for the integral, as we
use it, suggests to interprete the form of highest degree
$v$ as a volume form that is integrated over a non-commutative
space of Grassmann variables. The projection $Q$ is interpreted as the realm of integration whose
dimension is presisely the rank of $Q$.
The symbolic expression $v(\xi) \ F(\xi)$ is then a volume form with values in
the right module $\9E$, evaluated at $\xi$. 

Let $(e^i)_{i\in N}$ be a real orthonormal basis of $QH$. Each G-holomorphic function $F$ with
values in $\9E$ can be expanded with respect to this basis as $F(\xi)=\sum_{I\subset
N}F^I \xi_I$. Moreover, a form of highest degree $v\not=0$ can be expressed in terms of
this basis by $v=v_N\Lambda^N$.

This yields for the integral 
\begin{equation}
\int_Q v(\xi) \ F(\xi)=\sum_{I\subset N} F^Iv_I=v_{N} F^{N} \; .
\end{equation}
This shows that our definition of the
Grassmann integral is equivalent to the standard definition that can be found in the literature, see
for instance \cite{VladVol84,VladVol84b}. 

The following proposition lists some basic and well known properties of the Grassmann integral. We
also give here the proof, since our formalism (although equivalent) is a bit different from the one
that can be found in the literature. 

%%%%%%%%%%%%%%%%%%%%%%%%%%%%%%%%%%%%%%%%%%%%%%%%%%%%%%%%%%%%%%%%%%%%%%%%%%%%%%%%%%%%%%%%%%%%%%%%%%%%
\begin{prop}
\label{prop:GrInt}
Let $v$ be an non-zero form of highest degree with respect to a projection $Q$ of rank
$2n$. Then the Grassmann integral has the following properties:

\begin{enumerate}

\item 
The Grassmann integral is translationally invariant: For each G-holomorphic function $F$ with values
in $\9E$ the identity 
\begin{equation}
\int_Q v(\xi) \ F(\xi)= \int_Q {v}(\xi) \ F(\xi+\eta) 
\end{equation}
holds for all $\eta\in \PHSP{H}{J}{Q}_0$.

\item
For a G-holomorphic function $F$ with values in a right module $\9E$ the
identity 
\begin{equation}
\int_Q v \ F\cdot \lambda = \left[\int_Q {v} \ F \right]\lambda 
\end{equation}
holds for all $\lambda\in \GR{Q^\perp H}{Q^\perp J}$.

\item
Let $F$ be a G-holomorphic function with values in a right module $\9E$ and let
$T\mathpunct:\9E\to\9E'$ a right module homomorphism. For
each G-holomorphic function $F$ with values in $\9E$ the function $TF\mathpunct:\xi\mapsto
T(F(\xi))$ is G-holomorphic with values in $\9E'$ and the identity 
\begin{equation}
\int_Q v \ TF = T\left( \ \int_Q v F \right)
\end{equation}
holds. 

\end{enumerate}
\end{prop}
%%%%%%%%%%%%%%%%%%%%%%%%%%%%%%%%%%%%%%%%%%%%%%%%%%%%%%%%%%%%%%%%%%%%%%%%%%%%%%%%%%%%%%%%%%%%%%%%%%%%
\begin{proof}
Recall that we have introduced the Grassmann integral with help of the space of right module
homomorphisms $\9R_QF$. 
\begin{enumerate}
\item
For a right module homomorphism $R\in \9R_QF$, we obtain a right module homomorphism $\tau_\eta
R\in\9R_Q(\tau_\eta F)$ by putting $(\tau_\eta R)(a):=R(\exp(\eta)a)$.
For an operator $v$ of highest degree the integral of the translated function can be
calculated by $\int_Q v \ \tau_\eta F= \tau_\eta R(v)=R(\exp(\eta)v)$. Since
$v\Lambda(h)=0$ for all
$h\in QH$, it follows that  $v\eta=0$ for all $\eta\in\PHSP{H}{J}{Q}$. This implies
$\exp(\eta)v=v$ which yields the desired relation $\int_Q v  \ \tau_\eta F= R(v)=\int_Q v \
F$.

\item
Let $F$ be a G-holomorphic function with values in $\9E$. Then for $\lambda\in\GR{Q^\perp
H}{Q^\perp J}$ a right module homomorphism in $\9R_Q(F\cdot\lambda)$ is simply given by
$(R\cdot\lambda)(a)=R(\lambda a)$ with $R\in\9R_QF$. Thus we obtain for an operator $v$ of highest
degree
$\int_Q v F\cdot\lambda=R(\lambda v)=R(v\lambda)=R(v)\lambda=[\int_Q
v F] \lambda$.

\item
Let $T\mathpunct:\9E\to\9E'$ be a right module homomorphism, then it is obvious that
for a G-holomorphic function $F$ with values in $\9E$, the function
$TF\mathpunct:\xi\mapsto T(F(\xi))$ is G-holomorphic with values in $\9E'$. Since the identity
$(T\circ R)(\exp(\xi))=TF(\xi)$ is valid for all $R\in\9R_QF$, we conclude that $T\circ R\in
\9R_QTF$ which implies
for an operator $v$ of highest degree: $\int_Q v \ TF=(T\circ
R)(v)=T(R(v))=T(\int_Q v \ F)$.
\end{enumerate}
\end{proof}
%%%%%%%%%%%%%%%%%%%%%%%%%%%%%%%%%%%%%%%%%%%%%%%%%%%%%%%%%%%%%%%%%%%%%%%%%%%%%%%%%%%%%%%%%%%%%%%%%%%%

To treat multiple Grassmann integration, we have to say what is a G-holomorphic function in several
variables: A function on $\PHSP{H}{J}{Q}^n$ with values in $\9E$ is called G-holomorphic if for
each $j$ the function
\begin{equation}
\xi\mapsto F(\xi_1,\cdots,\xi_{j-1},\xi,\xi_{j+1},\cdots,\xi_n)
\end{equation}
is G-holomorphic. For our purpose, it is sufficient to consider the case $n=2$. For a G-holomorphic
function on $\PHSP{H}{J}{Q}^2$ with values in a right module $\9E$ we obtain two 
ordinary G-holomorphic functions $F_\eta$ and $F^\xi$ according to 
\begin{equation}
F^\xi(\eta):=F_\eta(\xi):=F(\eta,\xi) \; .
\end{equation}
The following proposition can be used for exchanging the order of multiple Grassmann integrations.

%%%%%%%%%%%%%%%%%%%%%%%%%%%%%%%%%%%%%%%%%%%%%%%%%%%%%%%%%%%%%%%%%%%%%%%%%%%%%%%%%%%%%%%%%%%%%%%%%%%%
\begin{prop}
\label{prop:MultInt}
Let $F$ be a G-holomorphic function on $\PHSP{H}{J}{Q}^2$ with values in $\9E$ and let 
$v,w$ be forms of highest degree with respect to the even rank projection $Q$. Then the functions
$F^v\mathpunct:\xi\mapsto \int_Q v \ F^\xi$ and $F_w\mathpunct:\eta\mapsto \int w \ F_\eta$ are
G-holomorphic and the order of integration can be exchanged: 
\begin{equation}
\int_{Q} v \ F_w=\int_{Q} w \ F^v
\end{equation}
\end{prop}
%%%%%%%%%%%%%%%%%%%%%%%%%%%%%%%%%%%%%%%%%%%%%%%%%%%%%%%%%%%%%%%%%%%%%%%%%%%%%%%%%%%%%%%%%%%%%%%%%%%%
\begin{proof}
Since $F$ is G-holomorphic in both variables, we obtain from the polynomial expansion that there
exists linear map $R\mathpunct:\GR{H}{J}^2\to\9E$ such that $R(a\lambda,b)=R(a,\theta(b))\lambda$
for $\lambda\in\GR{Q^\perp H}{Q^\perp J}_1$ and $R(a,b\lambda)=R(a,b)\lambda$ holds and that
fulfills the identity $F(\eta,\xi)=R(\8e^\eta,\8e^\xi)$.

For fixed $\eta$, the map $R_\eta\mathpunct:b\mapsto R(\8e^\eta,b)$ is a right module homomorphism
in $\9R_QF_\eta$. Thus we obtain for the partial integral $F_w(\eta)=\int_Q w\
F_\eta=R(\8e^\eta,w)$ which also shows that $F_w$ is G-holomorphic. Since $w$ is even, it follows
that the map $R_w\mathpunct:a\mapsto R(a,w)$ is a right module homomorphism in $\9R_Q F_w$ which
implies $\int_Q v \ F_w=R(v,w)$. By a similar argument one shows that $\int w \ F^v=R(v,w)$ which
implies the result.
\end{proof}
%%%%%%%%%%%%%%%%%%%%%%%%%%%%%%%%%%%%%%%%%%%%%%%%%%%%%%%%%%%%%%%%%%%%%%%%%%%%%%%%%%%%%%%%%%%%%%%%%%%%

%%%%%%%%%%%%%%%%%%%%%%%%%%%%%%%%%%%%%%%%%%%%%%%%%%%%%%%%%%%%%%%%%%%%%%%%%%%%%%%%%%%%%%%%%%%%%%%%%%%%
\subsection{Some useful lemmas for calculating Grassmann integrals}
%%%%%%%%%%%%%%%%%%%%%%%%%%%%%%%%%%%%%%%%%%%%%%%%%%%%%%%%%%%%%%%%%%%%%%%%%%%%%%%%%%%%%%%%%%%%%%%%%%%%
In order to calculate particular Grassmann integrals, a further interesting fact to mention is, that
the algebra of G-holomorphic functions is related to the Grassmann algebra $\GR{H}{J}$ itself. To
explain this, we observe that the Grassmann algebra (as a Banach space) $\GR{H}{J}$ is isomorphic to
the tensor product $\AFOCK{QH}\otimes\GR{Q^\perp H}{Q^\perp J}$. The canonical isomorphism is given
by identifying
\begin{equation}
\Lambda(f_1)\cdots \Lambda(f_n)\lambda\cong f_1\wedge\cdots\wedge f_n\otimes \lambda
\end{equation}
for $f_1,\cdots,f_n\in QH$ and $\lambda\in\GR{Q^\perp H}{Q^\perp J}$. This can be used
to introduce a rigging map on $\GR{H}{J}$ with values $\GR{Q^\perp H}{Q^\perp J}$.  By using the
induced Hilbert space structure on the Grassmann algebra $\GR{QH}{QJ}\cong \AFOCK{QH}$ the rigging
man is determined by 
\begin{equation}
\RIGG{\Lambda(f_1)\cdots \Lambda(f_n)\lambda}{\Lambda(h_1)\cdots\Lambda(h_n)\mu}{Q}
=\SPROD{f_1\wedge\cdots \wedge f_n}{h_1\wedge\cdots \wedge h_n}\lambda^\ADJ\mu \; .
\end{equation}
If we expand the operators in $\GR{H}{J}$ in a real orthonormal basis $(e^i)_{i\in N}$ of $QH$,
then the rigging map can be simply calculated as
\begin{equation}
\RIGG{a}{b}{Q}=\sum_{I\subset N}a_I^\ADJ b_I
\end{equation}
where $a=\sum_{I\subset N} \Lambda^I a_I$ and $b=\sum_{I\subset N} \Lambda^Ib_I$. We
associate to each operator $a\in\GR{H}{J}$ a G-holomorphic function according to 
\begin{equation}
E_Qa(\xi):=\RIGG{a^\ADJ}{\exp(\xi)}{Q}
\end{equation}
and we show the following useful lemma:

%%%%%%%%%%%%%%%%%%%%%%%%%%%%%%%%%%%%%%%%%%%%%%%%%%%%%%%%%%%%%%%%%%%%%%%%%%%%%%%%%%%%%%%%%%%%%%%%%%%%
\begin{lem}
\label{lem_app:iso}
The map $E_Q\mathpunct: \GR{H}{J}\to\RHOL{H}{J}{Q}$ is an algebra homomorphism. Moreover, for
each $a\in \GR{H}{J}$, the right module homomorphism
$R_Qa\mathpunct:b\mapsto\RIGG{a^\ADJ}{b}{Q}$ is contained in $\9R_Q E_Qa$.
\end{lem}
%%%%%%%%%%%%%%%%%%%%%%%%%%%%%%%%%%%%%%%%%%%%%%%%%%%%%%%%%%%%%%%%%%%%%%%%%%%%%%%%%%%%%%%%%%%%%%%%%%%%
\begin{proof}
Let $(e^i)_{i\in N}$ be a real orthonormal basis of $QH$. Then we expand the operator
$a\in\GR{H}{J}$ according to $a=\sum_{I\subset N} a^I\Lambda_I$ where the product 
$\Lambda_I=\Lambda(e^{i_1})\cdots \Lambda(e^{i_k})$ is ordered by increasing
indices $I=\{i_1<i_2<\cdots <i_k$ (indicated by a subscript $I$). Moreover, the exponential
$\exp(\xi)$ has an expansion $\exp(\xi)=\sum_{I\subset N} \Lambda^I\xi_I$ where the product 
$\Lambda^I=\Lambda(e^{i_k})\cdots \Lambda(e^{i_1})$ is ordered by decreasing
indices $I=\{i_1<i_2<\cdots <i_k$ (indicated by a superscript $I$). Since $a^\ADJ=\sum_{I\subset N}
\Lambda^I(a^I)^\ADJ$ we find
\begin{equation}
E_Qa(\xi):=\sum_{I\subset N} a^I \ \xi_I \; .
\end{equation}
For the product $E_Qa(\xi)E_Qb(\xi)$ we obtain the expansion
\begin{equation}
E_Qa(\xi)E_Qb(\xi)=\sum_{I,J\subset N} a^I \ \xi_I \ b^J \ \xi_J
=
\sum_{I\subset J\subset N} \epsilon_{IJ} \ a^I\theta^{|J\setminus I|}(b^{J\setminus I})\xi_J
\; ,
\end{equation}
where $\epsilon_{IJ}$ is the sign of the premutation $(I, J\setminus I)\to J$ which emerges from
the identity $\xi_I\xi_{J\setminus I}=\epsilon_{IJ} \xi_J$. On the
other hand, the product $ab$ has the following expansion with respect to the chosen basis: 
\begin{equation}
ab=\sum_{I,J\subset N} a^I\Lambda_I b^J \Lambda_J=
\sum_{I\subset J\subset N} \epsilon_{IJ} \ a^I\theta^{|J\setminus I|}(b^{J\setminus
I})\Lambda_J
\end{equation}
which implies the homomorphism property. Finally, we directly observe that
$E_Qa(\xi)=R_Qa(\8e^\xi)$ which concludes the proof.
\end{proof}
%%%%%%%%%%%%%%%%%%%%%%%%%%%%%%%%%%%%%%%%%%%%%%%%%%%%%%%%%%%%%%%%%%%%%%%%%%%%%%%%%%%%%%%%%%%%%%%%%%%%

%%%%%%%%%%%%%%%%%%%%%%%%%%%%%%%%%%%%%%%%%%%%%%%%%%%%%%%%%%%%%%%%%%%%%%%%%%%%%%%%%%%%%%%%%%%%%%%%%%%%
\begin{lem}
Let $F$ be a G-holomorphic function with values in a right module $\9E$ and let $v_1,v_2$ be two
operators of highest degree with respect to $Q$. Then the identity
\begin{equation}
\label{equ:delta}
\int_Q v_1(\xi) \ \int_Q v_2(\eta) \ F(\xi)
\8e^{\RIGG{\eta^\ADJ}{\xi-\zeta}{Q}}=
\SPROD{v_1^\ADJ}{v_2} F(\zeta) \; .
\end{equation}
is valid.
\end{lem}
%%%%%%%%%%%%%%%%%%%%%%%%%%%%%%%%%%%%%%%%%%%%%%%%%%%%%%%%%%%%%%%%%%%%%%%%%%%%%%%%%%%%%%%%%%%%%%%%%%%%
\begin{proof}
We take advantage of the fact that the bilinear form $(a,b)\mapsto\RIGG{a^\ADJ}{b}{Q}$ fulfills the
identity $\RIGG{a^\ADJ}{b}{Q}=\RIGG{b^\ADJ}{\theta_Q(a)}{Q}$ on the even subalgebra. Here
$\theta_Q$ is the automorphism $\theta_Q(\Lambda(f))=\Lambda((\I-2Q)f)$. From this we calculate the
following Grassmann integral: Let $v$ a form of highest degree with respect to a
projection $Q$ of even rank $n$. Then we calculate
\begin{equation}
\int_Q v(\xi) \8e^{\RIGG{\eta^\ADJ}{\xi}{Q}}=\RIGG{\8e^{\eta^\ADJ}}{v}{Q}
=\frac{1}{n!}\RIGG{v^\ADJ}{\eta^n}{Q}
\end{equation}
Furthermore, the value of the rigging map $\RIGG{a}{b}{Q}=\SPROD{a}{b}\I$ is a multiple of the
identity for all operators $a,b$ that belong to the subalgebra $\GR{QH}{QJ}$. As a consequence we
can calculate the double integral
\begin{equation}
\int_Q v_1(\eta) \ \int_Qv_2(\xi) \8e^{\RIGG{\eta^\ADJ}{\xi}{Q}}=
\int_Q v_1(\eta)\ \RIGG{v_2^\ADJ}{\8e^\eta}{Q}=
\SPROD{v_1^\ADJ}{v_2}\I \; .
\end{equation}
Finally we introduce the function $\delta_{v}$ that is given by
$\delta_{v}(\eta):=(n!)^{-1}\RIGG{v^\ADJ}{\eta^n}{Q}$ with $n=\8{dim}(Q)$. As we will see,
the G-holomorphic function $\delta_{v}$ plays the role of a $\delta$-function for Grassmann
integrals for the ``volume form'' $v$. To verify this, we observe 
that $\eta\delta_{v}(\eta)=0$ holds and we calculate for a G-holomorphic function $F$
the multiple integral
\begin{equation}
\int_Q v_1(\xi) \ \left[ \int_Q v_2(\eta) \ F(\xi)
\8e^{\RIGG{\eta^\ADJ}{\xi-\zeta}{Q}}\right] = \int_Q v_1(\xi) \ 
F(\xi)\delta_{v_2}(\xi-\zeta) \; .
\end{equation}
By taking advantage of the translation invariance of the Grassmann integral, it follows that 
\begin{equation}
\int_Q v_1(\xi) \ \left[ \int_Q v_2(\eta) \ F(\xi)
\8e^{\RIGG{\eta^\ADJ}{\xi-\zeta}{Q}}\right]=\int_Q v_1(\xi) \
\tau_\zeta F(\xi)\delta_{v_2}(\xi)
\end{equation}
holds. Since $\xi\delta_{v_2}(\xi)=0$, we find for each $R\in\9R_QF$
\begin{equation}
F(\xi)\delta_{v_2}(\xi)=R[\8e^\xi\delta_{v_2}(\xi)]=R(\I)=F(0) \; .
\end{equation}
Putting all these together implies the $\delta$-function formula:
\begin{equation}
\int_Q v_1(\xi) \ \left[ \int_Q v_2(\eta) \ F(\xi)
\8e^{\RIGG{\eta^\ADJ}{\xi-\zeta}{Q}}\right]
=
\tau_\zeta F(0) \ \int_Q v_1(\xi) \ \delta_{v_2}(\xi)
=
\SPROD{v_1^\ADJ}{v_2} F(\zeta) \; .
\end{equation}
\end{proof}
%%%%%%%%%%%%%%%%%%%%%%%%%%%%%%%%%%%%%%%%%%%%%%%%%%%%%%%%%%%%%%%%%%%%%%%%%%%%%%%%%%%%%%%%%%%%%%%%%%%%

%%%%%%%%%%%%%%%%%%%%%%%%%%%%%%%%%%%%%%%%%%%%%%%%%%%%%%%%%%%%%%%%%%%%%%%%%%%%%%%%%%%%%%%%%%%%%%%%%%%%
\subsection{Calculating Gaussian integrals}
%%%%%%%%%%%%%%%%%%%%%%%%%%%%%%%%%%%%%%%%%%%%%%%%%%%%%%%%%%%%%%%%%%%%%%%%%%%%%%%%%%%%%%%%%%%%%%%%%%%%
\label{subsec:calc-gauss}
Let $A$ be a bounded operator on $QH$, then there exists an operator $a\in \GR{QH}{QJ}$ such that 
\begin{equation}
\label{equ:quadratic}
\RIGG{\xi^\ADJ}{A\xi}{Q}=2\RIGG{a^\ADJ}{\exp(\xi)}{Q}
\end{equation}
where $a$ only depends on the anti-symmetric part $(A-JA^*J)/2$ of $A$. In particular, $a$ has
degree $2$ with respect to $Q$. To verify this, we recall
that the map $f\otimes h\mapsto \SPROD{Jf}{Ah}$ is a bilinear form on $QH$. Since $Q$  has finite
rank, it follows that
there exists a unique vector $\psi_A\in QH^{\otimes 2}$ such that $\SPROD{(J\otimes
J)\psi_A}{f\otimes h}=\SPROD{Jf}{Ah}$ holds. Moreover, we have
$\SPROD{Jh}{Af}=\SPROD{Jf}{JA^*Jh}=\SPROD{(J\otimes J)\7F\psi_A}{f\otimes h}$ where $\7F$ is the
flip operator that swaps the tensor product. The antisymmetric vector $a=(\7F\psi_A-\psi_A)/2$
can be identified with an operator in $\GR{H}{J}$ and we find 
\begin{equation}
\SPROD{a^\ADJ}{f\wedge h}=\frac{1}{2}\SPROD{Jf}{(JA^*J-A)h}
\end{equation}
According to the definition of the rigging map, we obtain for $\lambda,\mu\in\GR{Q^\perp
H}{Q^\perp J}_1$ and for $f,h\in QH$:
\begin{equation}
\RIGG{a^\ADJ}{\Lambda(f)\lambda
\Lambda(h)\mu}{Q}=\frac{1}{2}\SPROD{Jf}{(A-JA^*J)h}\lambda\mu
\end{equation}
which implies the identity (\ref{equ:quadratic}). In the following, we can therfore restrict the
consideration to operators on $QH$ that fulfill the condition $A=-JA^*J$. 

Since the operator $a$ has degree $2$ the operator $(n!)^{-1}a^n$ with $2n=\8{dim}(Q)$ is of
highest degree and therefore proportional to any other operator of highest degree. We
choose a selfadjoint operator of highest degree $v=v^\ADJ$ that
is normalized $\SPROD{v}{v}=1$. Then the identity 
\begin{equation}
(n!)^{-1}a^n=(n!)^{-1}\SPROD{v}{a^n}v 
\end{equation}
holds. If $a$ is related to an operator $A$ on $QB$ according to
$\SPROD{Jf}{Ah}=\SPROD{a^\ADJ}{f\wedge
h}$, then the scalar product $(n!)^{-1}\SPROD{v}{a^n}$ is the Pfaffian of $A$ with respect to
$v$and $J$ (see Subsection~\ref{subsec:main-results}):
\begin{equation}
\8{Pf}_{[v,J]}(A)=(n!)^{-1}\SPROD{v}{a^n} \; .
\end{equation}
This can now be applied calculate gaussian Grassmann integrals easily as stated by the next lemma:

%%%%%%%%%%%%%%%%%%%%%%%%%%%%%%%%%%%%%%%%%%%%%%%%%%%%%%%%%%%%%%%%%%%%%%%%%%%%%%%%%%%%%%%%%%%%%%%%%%%%
\begin{lem}
\label{lem_app:gauss_integral}
Let $A$ be an anti-symmetric operator on $QH$, i.e. $A=-JA^*J$, and let $v$ be a selfadjoint
normalized form of highest degree with respect to $Q$. Then the Gaussian integral identity 
\begin{equation}
\int_Q v(\xi) \
\8e^{\frac{1}{2}\RIGG{\xi^\ADJ}{A\xi}{Q}}=\8{Pf}_{[v,J]}(A)\I
\end{equation}
holds.
\end{lem}
%%%%%%%%%%%%%%%%%%%%%%%%%%%%%%%%%%%%%%%%%%%%%%%%%%%%%%%%%%%%%%%%%%%%%%%%%%%%%%%%%%%%%%%%%%%%%%%%%%%%
\begin{proof}
Let $a\in\GR{QH}{QJ}$ be the operator of degree $2$ that fulfills the identity
$\SPROD{a^\ADJ}{f\wedge h}=\SPROD{Jf}{Ah}$. Then we calculate the Gaussian integral with help of
Lemma~\ref{lem_app:iso} according to:
\begin{equation}
\int_Q v(\xi) \
\8e^{\frac{1}{2}\RIGG{\xi^\ADJ}{A\xi}{Q}}
=
\int_Qv(\xi) \ \8e^{\RIGG{a^\ADJ}{\8e^\xi}{Q}}=\int_Q v(\xi)\
\RIGG{\8e^{a^\ADJ}}{\8e^\xi}{Q}=\RIGG{\8e^{a^\ADJ}}{v}{Q} \; .
\end{equation} 
Since both operators $\8e^a$ and $v=v^\ADJ$ are even and contained in $\GR{QH}{QJ}$, we obtain for
the rigging map $\RIGG{\8e^{a^\ADJ}}{v}{Q}=\SPROD{\8e^{a^\ADJ}}{v}\I=\SPROD{v}{\8e^a}\I$. By
expanding the exponential $\8e^a$ only the contribution to the operator of highest degree
contribute to the scalar product. Therefore we obtain 
\begin{equation}
\SPROD{v}{\8e^a}=(n!)^{-1}\SPROD{v}{a^n} =\8{Pf}_v(a)=\8{Pf}_{[v,J]}(A)
\end{equation}
which proves the lemma.
\end{proof}
%%%%%%%%%%%%%%%%%%%%%%%%%%%%%%%%%%%%%%%%%%%%%%%%%%%%%%%%%%%%%%%%%%%%%%%%%%%%%%%%%%%%%%%%%%%%%%%%%%%%

The identity (\ref{equ:gaussian}) can now be shown by using the translation
invariance of the Grassmann integral. If $A=-JA^*J$ is invertible, we find 
\begin{equation}
\RIGG{(\xi-A^{-1}\eta)^\ADJ}{A(\xi-A^{-1}\eta}{Q}
=\RIGG{\xi^\ADJ}{A\xi}{Q}-\RIGG{\eta^\ADJ}{A^{-1}\eta}{Q}+2\RIGG{\eta^\ADJ}{\xi}{Q}
\end{equation}
where we have used the fact that $(A^{-1}\eta)^\ADJ=JA^{-1}J\eta^\ADJ$ holds. As a result we
obtain from the previous lemma
of highest degree with respect to $Q$. Then the Gaussian integral identity 
\begin{equation}
\int_Q v(\xi) \
\8e^{\frac{1}{2}\RIGG{\xi^\ADJ}{A\xi}{Q}+\RIGG{\eta^\ADJ}{\xi}{Q}}=\8{Pf}_{[v,J]}(A)\8e^{\frac{1}{
2 }
\RIGG{\eta^\ADJ}{A^{-1}\eta}{Q}} \; .
\end{equation}
In the case where $A$ does not fulfill the antisymmetry condition  $A=-JA^*J$ the Gaussian integral
formula is still valid by substituting $A$ on the right hand side by the antisymmetrized operator
$(A-JA^*J)/2$. 

%%%%%%%%%%%%%%%%%%%%%%%%%%%%%%%%%%%%%%%%%%%%%%%%%%%%%%%%%%%%%%%%%%%%%%%%%%%%%%%%%%%%%%%%%%%%%%%%%%%%

%%%%%%%%%%%%%%%%%%%%%%%%%%%%%%%%%%%%%%%%%%%%%%%%%%%%%%%%%%%%%%%%%%%%%%%%%%%%%%%%%%%%%%%%%%%%%%%%%%%%
\clearpage

%%%%%%%%%%%%%%%%%%%%%%%%%%%%%%%%%%%%%%%%%%%%%%%%%%%%%%%%%%%%%%%%%%%%%%%%%%%%%%%%%%%%%%%%%%%%%%%%%%%%
\begin{appendix}
%%%%%%%%%%%%%%%%%%%%%%%%%%%%%%%%%%%%%%%%%%%%%%%%%%%%%%%%%%%%%%%%%%%%%%%%%%%%%%%%%%%%%%%%%%%%%%%%%%%%

%%%%%%%%%%%%%%%%%%%%%%%%%%%%%%%%%%%%%%%%%%%%%%%%%%%%%%%%%%%%%%%%%%%%%%%%%%%%%%%%%%%%%%%%%%%%%%%%%%%%
\section{Proof of Theorem~\ref{thm:integral-rep}}
%%%%%%%%%%%%%%%%%%%%%%%%%%%%%%%%%%%%%%%%%%%%%%%%%%%%%%%%%%%%%%%%%%%%%%%%%%%%%%%%%%%%%%%%%%%%%%%%%%%%
\begin{proof}
In the first step, we show that each operator $A$ in the GAR algebra can be represented by a
Grassmann integral 
\begin{equation}
\label{eq:pre-decomp}
A=\int_Q v(\xi) \ {\1w}(\xi) \ f(\xi) 
\end{equation}
with a G-holomorphic function $f$ depending on $v$ and $A$. Let $(e^i)_{i\in N}$ be a real
orthonormal basis of $QH$ with $N=\{1,2,\cdots,\8{dim}(Q)\}$. Then
the Grassmann-Weyl operator can be expanded by 
\begin{equation}
\1w(\xi)=\sum_{I\subset N} B^I\xi_I
\end{equation}
where we have introduced the operators $B^I=G(e^{i_k})\cdots G(e^{i_1})$ for each ordered
subset $I=\{i_1<i_2<\cdots<i_k\}$. For each subset $K\subset N$ we introduce the G-holomorphic
function $b^K$ by $b^K(\xi):=\epsilon_{KN}\xi_{N\setminus J}$ where $\epsilon_{KN}$ is
determined by
the condition $\epsilon_{KN}\xi_K\xi_{X\setminus N}=\xi_N$. Then we conclude from the
polynomial expansion of the reduced Grassmann-Weyl operator $\1w(\xi)$ that
the identity
\begin{equation}
\begin{split}
\1w(\xi)b^K(\xi)=\sum_{I\subset N}  B^I \ \xi_I\xi_{N\setminus K}
=\sum_{I\subset K}\epsilon_{IKN}\epsilon_{KN} \ B^I \ \xi_{N\setminus (K\setminus I)} 
\end{split}
\end{equation}
holds with $\xi_I\xi_{N\setminus K}=\epsilon_{IKN}\xi_{N\setminus (K\setminus I)}$. Note that for
$I\not\subset K$ we have $\xi_I\xi_{N\setminus K}=0$. This implies for a form $v$ of highest
degree
\begin{equation}
\begin{split}
\int_Q v(\xi) \ {\1w}(\xi) \ b^{J}(\xi)
=
\int_Q v(\xi)  \ \sum_{I\subset K} \epsilon_{IKN}\epsilon_{KN} \ B^I \ \xi_{N\setminus
(K\setminus I)}=v_N \ B^K
\end{split}
\end{equation}
where we have used the identity $\epsilon_{KKN}=\epsilon_{KN}$. The operators $B^K$
form a basis of the fermionic part of the reduced GAR algebra. Thus a general operator can by
expanded as
$A=\sum_{I\subset N} B^I A_I$ with $A_I$ belonging to the Grassmann part. Thus the
G-holomorphic function
\begin{equation}
 f(\xi)=\sum_{I\subset N} v_N^{-1}\epsilon_{IN} \ \xi_{N\setminus I} \ A_I
\end{equation}
solves the identity (\ref{eq:pre-decomp}). 

We insert now this identity into the Fourier transform and obtain 
\begin{equation}
(\FOUR A)(\xi) = {\1w}(-\xi) \ \int_Q v(\eta) \int_Q v(\zeta) \ {\1w}(-\eta)
{\1w}(\zeta) \
f(\zeta) {\1w}(\eta) \exp(\RIGG{\eta^\ADJ}{\xi}{Q}) \; .
\end{equation}
By taking advantage of the fact that the Weyl operators belong to the even part of the GAR algebra,
we obtain from the Weyl relations
\begin{equation}
(\FOUR A)(\xi) = {\1w}(-\xi) \ \int_Q v(\eta) \int_Q v(\zeta) \ {\1w}(\zeta) \
f(\zeta)  \exp(\RIGG{\eta^\ADJ}{\xi-\zeta}{Q}) \; .
\end{equation}
Since $Q$ is even, the order of integration can be exchanged. Therefore, we get from the
$\delta$-function formula (\ref{equ:delta}) 
\begin{equation}
(\FOUR A)(\xi) = {\1w}(-\xi) \ \int_Q v(\zeta)  \ {\1w}(\zeta) \
f(\zeta) \int_Q v(\eta) \  \exp(\RIGG{\eta^\ADJ}{\xi-\zeta}{Q}) 
=\SPROD{v^\ADJ}{v} \ f(\xi)  \; .
\end{equation}
By choosing $v$ to be normalized and selfadjoint, implies the result.
\end{proof}
%%%%%%%%%%%%%%%%%%%%%%%%%%%%%%%%%%%%%%%%%%%%%%%%%%%%%%%%%%%%%%%%%%%%%%%%%%%%%%%%%%%%%%%%%%%%%%%%%%%%

%%%%%%%%%%%%%%%%%%%%%%%%%%%%%%%%%%%%%%%%%%%%%%%%%%%%%%%%%%%%%%%%%%%%%%%%%%%%%%%%%%%%%%%%%%%%%%%%%%%%
\section{Proof of Theorem~\ref{thm:convol-four}}
%%%%%%%%%%%%%%%%%%%%%%%%%%%%%%%%%%%%%%%%%%%%%%%%%%%%%%%%%%%%%%%%%%%%%%%%%%%%%%%%%%%%%%%%%%%%%%%%%%%%
\begin{proof}
We first consider the convolution of two G-holomorphic functions $f,f'$ which is given by 
\begin{equation}
(f\CON f')(\xi)=\int_Q v(\eta) \ f(\eta) \ f'(\xi-\eta) \; .
\end{equation}
Taking the Fourier transform yields
\begin{equation}
\FOUR (f\CON f')(\xi)=\int_Q v(\zeta) \int_Q v(\eta) \ f(\eta) \ f'(\zeta-\eta) \
\8e^{\RIGG{\xi^\ADJ}{\zeta}{Q}} \; .
\end{equation}
We make use of the translation invariance of the Grassmann integral (Proposition~\ref{prop:GrInt})
and the fact that the order of integration can be exchanged (Proposition~\ref{prop:MultInt}) which
implies 
\begin{equation}
\FOUR (f\CON f')(\xi)
=
\int_Q v(\eta) f(\eta) \8e^{\RIGG{\xi^\ADJ}{\eta}{Q}}  \ \int_Q v(\zeta) \ \ f'(\zeta) \
\8e^{\RIGG{\xi^\ADJ}{\zeta}{Q}}
=
\FOUR f(\xi)\FOUR f'(\xi) \; . 
\end{equation}

Let $\bS\varphi$ be a bounded right modul homomorphism from the GAR algebra into its Grassmann
part and let $A$ be an operator of the reduced Grassmann algebra. The the convolution of
$\bS\varphi$ and
$A$ is just given by $(\bS\varphi\CON A)(\xi)={\bS\varphi}(\alpha_\xi A)$. Therefore, the
Fourier
transform of this convolution is 
\begin{equation}
\FOUR(\bS\varphi\CON A)(\xi)
=
\int_Q v(\eta) \ {\bS\varphi}(\alpha_{-\eta} A)\ \8e^{\RIGG{\xi^\ADJ}{\eta}{Q}}
=
{\bS\varphi}\left(\int_Q v(\eta) \ \alpha_{-\eta} A\ \8e^{\RIGG{\xi^\ADJ}{\eta}{Q}}\right)
 \; .
\end{equation}
By Proposition~\ref{prop:GrInt}, the integration and the application by the right module
homomorphism can be exchanged. Keeping the definition of the Fourier transform of a GAR operator in
mind, we obtain: 
\begin{equation}
\FOUR(\bS\varphi\CON A)(\xi)={\bS\varphi}({\1w}(\xi)\FOUR A(\xi))
={\bS\varphi}({\1w}(\xi))\FOUR A(\xi)=\FOUR\bS\varphi(\xi)\FOUR A(\xi) \; .
\end{equation}
In the last step, we ahve used the property that $\FOUR A(\xi)$ is contained in the reduced
Grassmann part and that $\bS\varphi$ is right module homomrphism.

Finally, we consider the convolution of an operator $A$ in the reduced GAR algbera with a
G-holomorphic function $f$. The
Fourier transform is given by
\begin{equation}
\begin{split}
\FOUR(A\CON f)(\xi)
&=
\int_Q v(\eta) \alpha_{-\eta}\left(\int_Q v(\zeta) \ \alpha_\zeta A \
f(\zeta)\right)\8e^{\RIGG{\xi^\ADJ}{\eta}{Q}}
\\
&=
\int_Q v(\eta) \int_Q v(\zeta) \ \alpha_{\zeta-\eta} A \
f(\zeta) \ \8e^{\RIGG{\eta^\ADJ}{\xi}{Q}} \; .
\end{split}
\end{equation}
By using the translation invariance as well as the exchange rule for the order of integration, the
identity $\FOUR(A\CON f)=\FOUR A \FOUR f$ follows.
\end{proof}
%%%%%%%%%%%%%%%%%%%%%%%%%%%%%%%%%%%%%%%%%%%%%%%%%%%%%%%%%%%%%%%%%%%%%%%%%%%%%%%%%%%%%%%%%%%%%%%%%%%%

%%%%%%%%%%%%%%%%%%%%%%%%%%%%%%%%%%%%%%%%%%%%%%%%%%%%%%%%%%%%%%%%%%%%%%%%%%%%%%%%%%%%%%%%%%%%%%%%%%%%
\section{Proof of Theorem~\ref{thm:g-extension}}
%%%%%%%%%%%%%%%%%%%%%%%%%%%%%%%%%%%%%%%%%%%%%%%%%%%%%%%%%%%%%%%%%%%%%%%%%%%%%%%%%%%%%%%%%%%%%%%%%%%%
\begin{proof}
To start with, we first look at the left hand side of the equation (\ref{eq:char})
\begin{equation}
\FOUR{\boldsymbol\omega}_S(\xi)=\8e^{-\frac{1}{2}\RIGG{\xi^\ADJ}{S\xi}{Q}}
\end{equation}
that we have to show. For a real basis $(e^i)_{i\in N}$ of $QH$, the Grassmann-Weyl operator can be
expanded as a polynomial in the Grassmann variable $\xi$. The Fourier transform (characteristic
function) of the extended quasifree state $\bS\omega_S$ can be calculated by 
\begin{equation}
\FOUR\bS\omega_S(\xi)=\sum_{K\subset N} \omega_S(B^K) \ \xi_K \; ,
\end{equation}
where $B^K$ and $\xi_K$ are defined as within the proof of Theorem~\ref{thm:integral-rep} above.
We obtain a polynomial expansion with complex valued coefficients, given by the quasifree
expectation values $\omega_S(B^K)$. These expectation
values can be calculated by Wick's theorem
according to 
\begin{equation}
\omega_S(B^K)=\sum_{\Pi\in P_2(K)} \epsilon_{\Pi K} \prod_{I\in\Pi}\omega_S(B^I)
\end{equation}
where $P_2(K)$ is the set of all ordered partitions of $K$ into two-elementary subsets and
$\epsilon_{\Pi K}$ is the sign of the permutation $(I_1,\cdots,I_{k})\to K$ with
$\Pi=(I_1,\cdots,I_k)$. This yields for the full polynomial expansion:
\begin{equation}
{\bS\omega}_S({\1w}(\xi))=\sum_{K\subset N} \sum_{\Pi\in P_2(K)} \epsilon_{\Pi K}
\prod_{I\in\Pi}\omega_S(B^I)\ \xi_K \; .
\end{equation}
On the other hand, there exists an operator $a_S$ in $\GR{QH}{QJ}$ of degree 2 which is
determined by the condition $2\SPROD{a_S^\ADJ}{f\wedge h}=\SPROD{Jf}{Sh}$, $f,h\in QH$, and which
satisfies the identity
\begin{equation}
\RIGG{\8e^{a_S^\ADJ}}{\8e^\xi}{Q}=\8e^{-\frac{1}{2}\RIGG{\xi^\ADJ}{S\xi}{Q}} \; .
\end{equation}
If we expand the exponential of $a_S$ with respect to a real basis of $QH$, then we get, according
to the calculations we have done in the proof of Lemma~\ref{lem_app:iso}: 
\begin{equation}
\8e^{a_S}=\sum_{K\subset N} \sum_{\Pi\subset P_2(K)} \epsilon_{\Pi K} \prod_{I\in\Pi}a_S^I
\Lambda_K \; .
\end{equation}
This implies by using the expansion of the (reduced) rigging map: 
\begin{equation}
\8e^{-\frac{1}{2}\RIGG{\xi^\ADJ}{S\xi}{Q}}=\sum_{K\subset N} \sum_{\Pi\subset P_2(K)} \epsilon_{\Pi
K}
\prod_{I\in\Pi}a_S^I \  \xi_K \; .
\end{equation}
By construction, the identity $a_S^{\{i<j\}}=\SPROD{e^i}{Se^j}=\omega_S(G(e^i)G(e^j))$ holds for
all two-elementary ordered subsets $I=\{i<j\}$ whih implies (\ref{eq:char}).

Let $P$ be a basis projection on $QH$ and let $v$ be a selfadjoint normalized form of highest
degree in $\GR{QH}{QJ}$. Moreover, $E_P$ denotes the support projection of the
pure quasifree state $\omega_P$. We now have to show the identity
\begin{equation}
\FOUR_v({\bS\omega_P}\CON E_P)(\xi)=\8e^{-\RIGG{\xi^\ADJ}{P\xi}{Q}} \; .
\end{equation}
According to the definition of the convolution and the Fourier transform, we have to calculate the
Grassmann integral
\begin{equation}
\label{equ:proof-1}
\FOUR_v({\bS\omega_P}\CON E_P)(\xi)
=\int_Q v(\eta)
{\bS\omega}_P({\1w}(-\eta)E_P{\1w}(\eta))\8e^{\RIGG{\xi^\ADJ}{\eta}{Q}} \; .
\end{equation}
Since $\omega_P$ is pure, the suppost projection $E_P$ has rank one and the identity
$\omega_P(AE_PB)=\omega_P(A)\omega_P(B)$ holds for $A,B$ in the fermionic part. This property is
lifted to the G-extension, i.e. ${\bS\omega}_P(AE_PB)={\bS\omega}_P(A){\bS\omega}_P(B)$. Since the
equivalence class mapping $\EQ{\cdot}{Q}$ is a *-algebra homomorphism, we obtain for the lift
${\bS\omega}_P$ to the reduced GAR algebra that
${\bS\omega}_P(AE_PB)={\bS\omega}_P(A){\bS\omega}_P(B)$ holds for all operators $A,B$
in the reduced GAR algebra. To verify
this we choose operators of the form $A\lambda$ and $B\mu$, where $A,B$ belong to the fermionic
part and $\lambda\in\GR{Q^\perp H}{Q^\perp J}_q$, $\mu$ belong to the Grassmann part
of the GAR algebra. Then we calculate
\begin{equation}
\begin{split}
{\bS\omega}_P(A\lambda E_P B\mu)&={\bS\omega}_P(AE_P\theta^q(B)\lambda\mu)\\
&=\omega_P(AE_P\theta^q(B))\lambda\mu
\\
&=\omega_P(A)\omega_P(B)\lambda\mu
\\
&={\bS\omega}_P(A\lambda){\bS\omega}_P(B\mu) \; .
\end{split}
\end{equation}
By using the identity (\ref{eq:char}), which we just have proven above, we get 
\begin{equation}
{\bS\omega}_P({\1w}(-\eta)E_P{\1w}(\eta))={\bS\omega_P}(\1w(\eta))^2=\8e^{-\RIGG{
\eta^\ADJ } { P\eta } { Q } }
\; .
\end{equation}
By inserting this into Equation (\ref{equ:proof-1}), it remains to calculate the Gaussian integral
with help of Lemma~\ref{lem_app:gauss_integral} and the disussion thereafter in
Subsection~\ref{subsec:calc-gauss}:
\begin{equation}
\FOUR_v({\bS\omega_P}\CON E_P)(\xi)=\int_Q
v(\eta)
\8e^{-\RIGG{\eta^\ADJ}{P\eta}{Q}+\RIGG{\xi^\ADJ}{\eta}{Q}}
=\8{Pf}_{[v,J]}\left(JPJ-P\right)\8e^{\frac{1}{2}\RIGG{\xi^\ADJ}{(JPJ-P)^{-1}\xi}{Q}} \; .
\end{equation}
Since $P$ is a basis projection, the operator $JPJ-P=\I-2P$ is a reflection and we have
$\8{Pf}_{[v,J]}(\I-2P)=\pm 1$. Moreover, we conclude
$\RIGG{\xi^\ADJ}{(JPJ-P)^{-1}\xi}{Q}=\RIGG{\xi^\ADJ}{(\I-2P)\xi}{Q}=-2\RIGG{\xi^\ADJ}{P\xi}{Q}$. 
This yields the desired result.
\end{proof}
%%%%%%%%%%%%%%%%%%%%%%%%%%%%%%%%%%%%%%%%%%%%%%%%%%%%%%%%%%%%%%%%%%%%%%%%%%%%%%%%%%%%%%%%%%%%%%%%%%%%

\end{appendix}
%%%%%%%%%%%%%%%%%%%%%%%%%%%%%%%%%%%%%%%%%%%%%%%%%%%%%%%%%%%%%%%%%%%%%%%%%%%%%%%%%%%%%%%%%%%%%%%%%%%%

%%%%%%%%%%%%%%%%%%%%%%%%%%%%%%%%%%%%%%%%%%%%%%%%%%%%%%%%%%%%%%%%%%%%%%%%%%%%%%%%%%%%%%%%%%%%%%%%%%%%
\end{document}